\documentclass{article}
\usepackage[utf8]{inputenc}

\title{Tight Lower Bound on Equivalence Testing in Conditional Sampling Model}

\author{Diptarka Chakraborty\footnote{National University of Singapore, Singapore. Supported in part by an MoE AcRF Tier 2 grant (MOE-T2EP20221-0009) and Google South \& South-East Asia Research Award. Email: diptarka@comp.nus.edu.sg}\and Sourav Chakraborty \footnote{Indian Statistical Institute, Kolkata. Email: sourav@isical.ac.in} \and Gunjan Kumar \footnote{National University of Singapore, Singapore. Supported in part by  National Research Foundation Singapore under its NRF Fellowship Programme[NRF-NRFFAI1-2019-0004 ]. Email: dcsgunj@nus.edu.sg}}
\date{}

\usepackage{amsthm,amsmath,amssymb}
\usepackage{thmtools}
\usepackage{bm}
\newcommand{\E}{\mathbb{E}}

\newcommand{\mcA}{\mathcal{A}}

\usepackage{algorithm}
\usepackage[noend]{algpseudocode}
\usepackage{fullpage}
\usepackage{todonotes}
\usepackage{tikz}
\newcommand{\YES}{\ensuremath{\mathsf{YES}}}
\newcommand{\NO}{\ensuremath{\mathsf{NO}}}
\newcommand{\Bad}{\ensuremath{\mathsf{Bad}}}
\newcommand{\Good}{\ensuremath{\mathsf{Good}}}
\newcommand{\G}{\ensuremath{\mathsf{G}}}
\newcommand{\K}{\ensuremath{\mathsf{K}}}
\newcommand{\tv}{\ensuremath{\mathsf{tv}}}
\newcommand{\dtv}{\ensuremath{\mathsf{d_{TV}}}}
\newcommand{\IY}{\ensuremath{\mathcal{I}_{\mathsf{YES}}}}
\newcommand{\IN}{\ensuremath{\mathcal{I}_{\mathsf{NO}}}}
\newcommand{\IE}{\ensuremath{\mathcal{I}_{\mathsf{EVEN}}}}
\newcommand{\IO}{\ensuremath{\mathcal{I}_{\mathsf{ODD}}}}
\newcommand{\Alg}{\ensuremath{\mathcal{A}}}

\newcommand{\COND}{\ensuremath{\mathsf{COND}}}
\newcommand{\At}{\ensuremath{\mathsf{At}}}
\newcommand{\SAMP}{\ensuremath{\mathsf{SAMP}}}
\newcommand{\mcB}{\ensuremath{\mathcal{B}}}
\newcommand{\Ratio}{\ensuremath{\mathsf{WCOND}}}
\DeclareMathOperator{\poly}{poly}

\renewcommand{\epsilon}{\varepsilon}

\newtheorem{theorem}{Theorem}[section]
\newtheorem{claim}[theorem]{Claim}
\newtheorem{lemma}[theorem]{Lemma}
\newtheorem{definition}[theorem]{Definition}
\newtheorem{corollary}[theorem]{Corollary}
\newtheorem{observation}[theorem]{Observation}
\begin{document}

\maketitle
\thispagestyle{empty}
\setcounter{page}{0}
\begin{abstract}
    We study the equivalence testing problem where the goal is to determine if the given two unknown distributions on $[n]$ are equal or $\epsilon$-far in the total variation distance in the conditional sampling model (CFGM, SICOMP16; CRS, SICOMP15)  wherein a tester can get a sample from the distribution conditioned on any subset. Equivalence testing is a central problem in distribution testing, and there has been a plethora of work on this topic in various sampling models.
    
   Despite significant efforts over the years, there remains a gap in the current best-known upper bound of $\Tilde{O}(\log \log n)$ [FJOPS, COLT 2015] and lower bound of $\Omega(\sqrt{\log \log n})$[ACK, RANDOM 2015, Theory of Computing 2018].
   Closing this gap has been repeatedly posed as an open problem (listed as problems 66 and 87 at sublinear.info). In this paper, we completely resolve the query complexity of this problem by showing a lower bound of $\Tilde{\Omega}(\log \log n)$. For that purpose, we develop a novel and generic proof technique that enables us to break the $\sqrt{\log \log n}$ barrier, not only for the equivalence testing problem but also for other distribution testing problems, such as uniblock property.
\end{abstract}
\newpage
\section{Introduction}

Probability distributions play a central role in modern data science, and consequently, the past few years have witnessed sustained interest from theoreticians and practitioners alike in the broad field of distribution testing, wherein the central object of study is probability distribution(s). In this work, we focus on the discrete distributions over the domain of size $n$. Since the value of $n$ is often too large for distributions of interest, it is  impractical to specify such a distribution. Therefore, one is often interested in quantifying complexity through the lens of queries to the distributions. The goal in this scenario is to check whether the input distribution(s) has some particular property or is (are) ``$\epsilon$-far" from satisfying that property, and doing all these while trying to reduce the number of queries made to the distribution(s).  

Initial studies in distribution property testing focused on the model ($\SAMP$) wherein one can only sample from the given distribution(s). The $\SAMP$ model was discovered to be too weak, as evidenced by strong lower bounds of the form $\Omega(n^{1-c})$, for some constant $c \ge 0$, for testing some of the most exciting properties. Such strong lower bounds necessitated the need to allow more powerful queries, and over the past decade, several models have been proposed. Among such proposals, the conditional sampling model ($ \COND $) -- that allows drawing samples from the input distribution(s) conditioned on any arbitrary subset -- is the most well-studied model in theory as well as in practice. From a theoretical perspective, various other distribution testing problems have been studied under the $\COND$ model~\cite{falahatgar2015faster, kamath2019anaconda, narayanan2021tolerant} and certain variants of it like \emph{subcube conditioning model}~\cite{bhattacharyya2018property, canonne2021random, chen2021learning}. Furthermore, the $\COND$ model and its variants have recently found applications in the areas like formal methods and machine learning (e.g.,~\cite{chakraborty2019testing, meel2020testing, golia2022scalable}).

In this paper, we focus on the equivalence testing problem, one of the central problems in the field of distribution testing. In particular, we want to determine whether two distributions $D_1$ and $D_2$ are equal or $\epsilon$-far from each other in the total variation distance under the $\COND$ model. Equivalence testing is arguably the most celebrated problem in distribution testing. 
In the $\SAMP$ model, the equivalence testing is well understood, and its query complexity is $\Theta(\max(n^{2/3}/\epsilon^{4/3},\sqrt{n}/\epsilon^2))$~\cite{chan2014optimal,batu2013testing,valiant2011testing}. 

Analyzing the complexity of equivalence testing has turned out to be more challenging in the context of the $\COND$ model. Despite significant efforts over the years, there remains a quadratic gap between the current best-known upper bound of $\Tilde{O}(\log \log n)$~\cite{falahatgar2015faster} and the lower bound of $\Omega(\sqrt{\log \log n})$~\cite{acharya2018chasm}. The challenge of closing the gap between the lower and upper bounds has been a recurring open problem, which has been raised multiple times at various workshops and conferences, including the 2014 Bertinoro Workshop on Sublinear Algorithms and the FOCS 2017 Frontiers of Distribution Testing. This problem has also been discussed on open problem forums such as sublinear.info (listed as problems 66 and 87).

The main difficulty in bridging the gap stems from the limitations inherent in the current approach that was used to establish the lower bound of $\Omega(\sqrt{\log \log n})$. This limitation was best highlighted by the authors of~\cite{acharya2018chasm}:
``\emph{There appear to be conceptual barriers to strengthening our result, which would require new ideas}". 

The primary contribution of this paper is to develop a novel and generic technique that overcomes the limitations of previous proof techniques and enables us to go past the $\Omega(\sqrt{\log \log n})$ lower bound barrier for not just equivalence testing but for several other problems in distribution testing.

\subsection{Our Lower Bound Result on Equivalence Testing}

The main contribution of this paper is to prove an (almost) tight lower bound on the query complexity in the $\COND$ query model for the equivalence testing of distributions (see Definition~\ref{def:equivalence} for the formal definition). In the $\COND$ query model, the tester can specify a subset $A \subseteq [n]$ and then samples each $j \in A$ according to the distribution $D$ conditioned on the set $A$, i.e., with probability $D(j)/D(A)$ (see Definition~\ref{def:cond}). 
Let us now state our main result.

\begin{theorem}
    \label{thm:mainlb}
    Any (randomized) adaptive tester for testing equivalence between two distributions over $[n]$ must make  $\Tilde{\Omega}(\log \log n)$ $\COND$ queries. (The tilde hides a $\poly(\log\log\log n)$ factor.)
\end{theorem}

We prove this result by introducing a weaker query model called the $\Ratio$ query model that is easier to analyze, then proving the query lower bound in that weaker model, and finally showing that this weaker model is (roughly) equivalent to the $\COND$ model at least for the equivalence testing problem. We discuss this approach in more detail in Section~\ref{sec:overview}. We believe that our proof technique is very generic and can be used for other distribution testing problems in the $\COND$ query model. Indeed, the same technique helps us establish a query complexity lower bound for another class of problems - testing label invariant properties.

\subsection{Our Lower Bound Result on testing Label Invariant properties}

A property of a distribution that is invariant under relabeling of the universe (on which the distribution is defined) is called \emph{label-invariant}. For a label-invariant property $\mathcal{P}$, the goal is to check if a given distribution satisfies the property $\mathcal{P}$ or $\epsilon$-far (in the total variation distance) from satisfying the property $\mathcal{P}$.
One crucial difference between the problem of testing equivalence and the problem of testing any label invariant property is that in the former, two distributions are given as input, and both have to be accessed using queries, while in the latter, the input is only one distribution that has to be accessed using queries.

CFGM~\cite{chakraborty2016power}  showed a universal tester to test any label invariant property in $\poly(\log n,1/\epsilon)$ queries. More precisely, there exists a tester that, for \textit{any} label invariant property $\mathcal{P}$, given $\COND$-access to a distribution $D$ as input, makes   $\poly(\log n)$  queries and with high probability,  returns ACCEPT if $D$ satisfies the property $\mathcal{P}$  and REJECT if $D$ is $\epsilon$-far in total variation distance from any distribution having property $\mathcal{P}$. They also defined a label invariant property, called \emph{even uniblock} property, and showed a lower bound of $\Omega(\sqrt{\log \log n})$ on the query complexity. The significance of this lower bound is that now we cannot hope to show a universal tester that can test any label invariant property in $o(\sqrt{\log \log n})$ queries.

In this paper, we improve this lower bound to $\Tilde{\Omega}(\log \log n)$.

\begin{theorem}\label{thm:main-odd-even}
There exists a label invariant property such that any (randomized) adaptive tester for that property must make $\Tilde{\Omega}(\log \log n)$ $\COND$-queries.
\end{theorem}

We follow the same approach that we take for showing a lower bound for the equivalence testing - via the $\Ratio$ query model. While this improves the lower bound by a quadratic factor, the exact bound on the query complexity of label-invariant properties remains unknown.

\subsection{Related Work}
The equivalence testing problem has been studied extensively~\cite{batu2000testing,valiant2011testing,chan2014optimal} and a tight bound of  $\Theta(\max(n^{2/3}/\epsilon^{4/3},\sqrt{n}/\epsilon^2))$ on the query complexity is known in the basic $\SAMP$ model.
There has been considerable recent interest in several alternative powerful query models that allow tremendous savings in the number of required samples. \cite{canonne2015testing} showed an upper bound of $O(\log^5 n / \epsilon^4)$ in the $\COND$ model which was subsequently improved to $\Tilde{O}(\frac{\log \log n}{\epsilon^5})$ by~\cite{falahatgar2015faster}. On the other hand, \cite{acharya2018chasm} showed that $\Omega(\sqrt{\log \log n})$ queries are necessary for the $\COND$ model, leaving a quadratic gap (in the dependence on the domain size) between the upper and lower bound, which we settle in this work. In terms of $\epsilon$, a lower bound of $\Omega(1/\epsilon^2)$ is known~\cite{canonne2014testing}, and finding the true dependency on $\epsilon$ in the query complexity is an exciting open problem.

Onak and Sun~\cite{onak2018probability} considered a natural extension of $\SAMP$ model called the \emph{probability-revealing sample} model (in short, \textsc{PR-SAMP}), wherein in addition to returning a sample, the oracle also returns the exact probability of the sample. An even more powerful model of interest is \textsc{DUAL}~\cite{canonne2014testing,batu2002complexity} where we have access to two oracles -- $\SAMP$ that, as mentioned earlier, provides a sample from the distribution, and \textsc{EVAL} that returns the exact probability of any specified element\footnote{Note that the \textsc{PR-SAMP} oracle can only provide the probability of the sampled element, whereas the \textsc{EVAL} oracle can give the probability of any arbitrary specified element.}. Cannone and Rubinfield~\cite{canonne2014testing} showed that $O(1/\epsilon)$ queries are necessary and sufficient for the equivalence testing in the \textsc{DUAL} model. 

 Another fundamental distribution testing problem is the \emph{support size estimation} problem, where given a distribution $D$, the goal is to estimate the support size $\left|\{x \mid D(x) > 0\}\right|$. Very recently, Chakraborty, Kumar, and Meel (CKM)~\cite{CKM23} showed a tight lower bound of $\Omega(\log \log n)$ on the query complexity for the support-size estimation problem in the {\COND} model. However, the lower bound for the support size estimation does not imply any lower bound on the equivalence testing. Two distributions can have the same support size but can be arbitrarily far apart. Similarly, two distributions can be arbitrarily close to each other, but their support sizes can differ by an arbitrarily large value. For instance, consider distribution $P = (1,0, \cdots ,0)$ and $Q = (1 - \epsilon, \epsilon/K, \cdots ,\epsilon/K)$ where $\epsilon$ is any arbitrary small value and $K$ is an arbitrarily large integer. It is straightforward to see that the total variation distance between $P$ and $Q$ is $2 \epsilon$ (a small value), but their support size differs by a factor of $K$ (an arbitrarily large number). Furthermore, the proof technique of CKM does not yield any lower bound for the equivalence testing. It is also worth noting that the hard instance of CKM allows inverse exponential probability mass on an element in the distribution to achieve the lower bound, which is most often not allowed in the distribution testing setup (where we are concerned with additive error). We want to emphasize that by applying our new proof technique, we obtain a similar lower bound for the support size estimation problem, which, more importantly, holds even when it is promised that all the elements with non-zero probability have mass at least $1/n$ and thus strengthen the lower bound result of CKM (see Section~\ref{sec:support} for a brief discussion on this). 

\paragraph{Organisation of the paper.} In Section~\ref{sec:notations}, we provide the necessary notations, definitions, and important theorems that we use in the rest of the paper. In Section~\ref{sec:overview}, we provide the technical overview of our results. We first discuss the previous techniques and how we improve upon them for our results and then present a sketch of our proof technique in this section. In Section~\ref{sec:main} and in Section~\ref{sec:label}, we present the proofs of Theorem~\ref{thm:mainlb} and Theorem~\ref{thm:main-odd-even} respectively. We conclude with a discussion on the limitation of our technique in Section~\ref{sec:conclusion}.

\section{Notations and Preliminaries}
\label{sec:notations}

 
In this paper, we assume that all the distributions are defined over the set $[n]$. For any distribution $D$ over $[n]$ and for any $s\in [n]$, we will denote by $D(s)$ the probability weight of $s$ according to the distribution $D$.  Similarly, for any $A\subseteq [n]$, we denote by $D(A)$ the probability mass of the set $A$ according to the distribution $D$. In other words, $$D(A) = \sum_{a\in A} D(a).$$
The \emph{total variation distance} between two distributions $D_1$ and $D_2$ over $[n]$, denoted by $\dtv(D_1,D_2)$, is defined as $\frac{1}{2}\sum_{s\in [n]}|D_1(s) - D_2(s)|$.  If the $\dtv$ distance between two distributions is $0$, then they are \emph{equal}. 

\begin{definition}[Equivalence Testing] \label{def:equivalence}
    The \emph{equivalence testing problem} is that, given sample access to two (unknown) distributions $D_1$ and $D_2$ over $[n]$, and an $\epsilon > 0$, decide whether
    \begin{itemize}
        \item $\YES:$ $D_1$ and $D_2$ are equal, or
        \item $\NO:$ $\dtv(D_1,D_2) \ge \epsilon$ 
    \end{itemize}
    while drawing as few samples as possible. 
\end{definition}
Throughout this paper, for the purpose of proving the lower bound, we consider $\epsilon =1/4$. 

The problem of equivalence testing of distributions is one of the most fundamental problems in statistics and property testing and has been studied under various sampling models. 
In this paper, our main sampling model is the $\COND$ model. To define this formally, we first need to define the conditional distribution. 

\begin{definition}
    For a distribution $D$ over $[n]$ and a subset $A\subseteq [n]$, the conditional distribution over $A$, denoted by $D|_A$, is defined as the distribution over $A$ where for each $a\in A$ the probability mass is set to be $D(a)/D(A)$ if $D(A) > 0$, and $1/|A|$ if $D(A) =0$.
\end{definition}

\begin{definition}[$\COND$ Query Model]
\label{def:cond}
    In the $\COND$ query model (or simply $\COND$ model), the sampling algorithm/tester specifies a subset $A\subseteq [n]$ and draws a sample according to the conditional distribution $D|_A$. We denote such a conditional query by $\COND_D(A)$. Note that, in the case of an adaptive algorithm, at any point in time, the subset $A$ may depend on the samples the algorithm has previously obtained. 
\end{definition}
 In this paper, we deal with adaptive algorithms.

\subsection*{Hypergeometric Distribution}
The \emph{hypergeometric distribution$(n, K, N)$} is a probability distribution that describes the number of successes (drawn item has a specified feature) when $n$ items are drawn without replacement from a population of size $N$ containing $K$ objects with that feature. Note that when items are drawn with replacement, the distribution becomes $Binomial(n, K, N)$. 

If $X \sim Hypergeometric(n,K,N)$ then like binomial distribution, we have $\E[X] = \frac{nK}{N}$. Further, Chernoff bound for hypergeometric distribution holds similar to  binomial distributions.
\begin{lemma}
\label{lem:chernoff}
    Let $X \sim Hypergeometric(n,K,N)$ then  $\mu = \E[X] = \frac{nK}{N}$ and 
    \[
    \Pr\left[|X-\mu| \ge \lambda \mu\right] < 2exp\left(-\frac{\lambda^2 \mu }{3}\right), \quad \text{for any }  0 \le \lambda \le 1.
    \]
\end{lemma}

\subsection*{Extension of Yao's Lemma}

One useful tool for proving the lower bound on the query complexity of various problems is the extension of Yao's lemma (formally proved in~\cite{fischer04}). Let $\IY$ and $\IN$ be two distributions over the $\YES$-instances and $\NO$-instances respectively. We use the notation $x \in_R \IY$ (resp., $x \in_R \IN$) to denote that $x$ is drawn uniformly at random from $\IY$ (resp., $\IN$). Let a single query return an element from the set $[n]$. For a deterministic query algorithm $\Alg$ that makes $q$ adaptive queries\footnote{We can assume without loss of generality that an adaptive algorithm that makes at most $q$ queries actually makes exactly $q$ queries.} note that all the answers to the $q$ queries is an element of $[n]^q$. From now on, we consider a tiny constant $\delta = 1/100$.



\begin{theorem}[\cite{fischer04}]
\label{thm:extyao}
If for a deterministic algorithm $\Alg$ that makes $q$ adaptive queries to test a property $\mathcal{P}$, and for an event $\Bad(\Alg, x)$ (that depends on the algorithm $\Alg$ and the input) the following holds 
\begin{enumerate}
\item $\Pr_{x\in_R \IY}[\Bad(\Alg, x)] + \Pr_{x\in_R \IN}[\Bad(\Alg, x)] \leq \delta/2$
\item For all $\sigma\in [n]^q$,
\begin{align*}
\Pr_{x\in_R \IY} \Big[\mbox{ the answers to the $q$} & \mbox{ queries made by $\Alg(x)$ is $\sigma$} \mid \overline{\Bad(\Alg, x)}\Big] \\
\leq & \frac{3}{2}\Pr_{x\in_R \IN} \left[\mbox{ the answers to the $q$ queries made by $\Alg(x)$ is $\sigma$ } \mid \overline{\Bad(\Alg, x)} \right],\\
\end{align*}
\end{enumerate}
then the adaptive query complexity of the property $\mathcal{P}$ is $\Omega(q)$.
\end{theorem}

\section{Technical Overview}\label{sec:overview}

\subsection{Previous Approach}
\label{sec:prev}
In the context of distribution testing, the $\COND$ model is known to be significantly more powerful than the $\SAMP$ model. Consequently, establishing a lower bound for the $\COND$ model is an immense challenge. One of the reasons it is so difficult to capture a tester's power in this model is that it allows the conditioning of arbitrary-sized sets in an adaptive manner. To overcome this challenge, ~\cite{chakraborty2016power} introduced the concept of \emph{core-adaptive testers} for label-invariant properties. Roughly speaking, these testers do not consider the samples' labels into account when making decisions; instead, they rely on relations between the samples, such as whether two samples are the same or different. Quite surprisingly, they showed that the class of core-adaptive testers is as powerful as general testers in terms of testing label invariant properties. 

Later,~\cite{acharya2018chasm} further built upon this idea by considering two classes of pairs of distributions. The first class consists of pairs in which both distributions are identical ($\YES$-instance), while the second class consists of pairs that are far apart in terms of the total variation distance ($\NO$-instance). The authors proved that any core-adaptive tester must make at least $q = \Omega(\sqrt{\log \log n})$ queries; otherwise, the distributions from the $\YES$ and $\NO$ instances become indistinguishable from the tester's point of view. The idea of the core-adaptive tester essentially helps in upper bounding the size of the corresponding decision tree $R$ by $2^{O(q^2)}$. \cite{acharya2018chasm} argued that to distinguish between the $\YES$ and $\NO$ instances, the number of nodes present in $R$ must be $\Omega(\sqrt{\log n})$, and as a consequence, $q \ge \Omega(\sqrt{\log \log n})$. While this lower bound does not match the current best upper bound, it is optimal with respect to the proof technique, which is also emphasized in the survey~\cite{canonne2020survey} as:

``\emph{The fact that both the lower bounds are similar is not a coincidence, but rather inherent to the technique used. Indeed, the core adaptive tester approaches both proofs rely on cannot get past this  
$\sqrt{\log \log n}$ barrier, which derives from the size of the decision tree representing the tester (namely, $2^{O(q^2)}$ for a $q$-query tester)}". 

Our primary contribution lies in developing a novel and generic technique that overcomes the limitations of previous proof techniques and enables us to break the $\sqrt{\log \log n}$ barrier in the $\COND$ model. Our technique does not only apply to the equivalence testing problem, but we believe it can be applied to various other problems as well. For instance, we show that it also provides $\Tilde{\Omega}(\log \log n)$ lower bound to the problem of testing another label-invariant property called the \emph{even-uniblock property}, introduced by~\cite{chakraborty2016power}. 






\subsection{Core-adaptive testers}
\label{sec:core-adapt}
As in previous approaches, core-adaptive testers also play a crucial role in our proof technique. Thus let us start by describing the core-adaptive testers. Any general algorithm/tester for equivalence testing between two distributions $D_1$ and $D_2$ that makes at most $q$ $\COND$ queries, at any step $1 \le i \le q$, chooses $k \in \{1,2\}$, a set $A_i \subseteq [n]$ and places the conditional query $\COND_{D_k}(A_i)$, and then receives a sample $s_i \in A_i$ drawn according to the conditional distribution $D_k|_{A_i}$. Note that if the algorithm queries both distributions $D_1$ and $D_2$ on a set $A$, then we count it as two separate queries.

In~\cite{chakraborty2016power}, it is shown that without loss of generality, a general algorithm/tester for a label invariant property (such as equivalence testing) can be assumed to belong to a smaller class of testers called \emph{core-adaptive testers}. To formally define these testers, we first give a few definitions.
\begin{definition}[Atom]
    Given a family of sets $\mcA = \{A_1,\dots,A_i\},$ the \emph{atoms} generated by  $\mcA$, denoted by $\At(\mcA)$, are (at most) $2^i$ distinct sets of the form $\cap_{j=1}^{i}C_j$ where $C_j \in \{A_j,[n]\setminus A_j\}$.
\end{definition}
 For example, if $i =2$, then $\At(A_1, A_2) = \{A_1\cap A_2, A_1\setminus A_2, A_2\setminus A_1, \overline{A_1\cup A_2}\}$. Given a sequence of query-sample pairs  $((A_1,s_1),\dots,(A_i,s_i))$,  all the label invariant information  about the sample $s_i$ can be captured by the \emph{configuration} of $s_i$, defined below.


\begin{definition}[Configuration of $s_i$]
\label{def:configuration}
    Given a sequence of query-sample pairs $((A_1,s_1),\dots, (A_i,s_i))$, a \emph{configuration} of $s_i$ with respect to $((A_1,s_1),\dots, (A_i,s_i))$, denoted by $c_i$, consists of $2(i-1)$ bits, indicating for each $1 \le \ell < i$ whether
    \begin{enumerate}
        \item \label{step:confg-1} $s_i = s_\ell$ or $s_i \neq s_\ell$, and
        \item \label{step:confg-2} $s_i \in A_\ell$ or not.
    \end{enumerate}
\end{definition}
Note that a configuration of $s_i$ contains all the label-invariant information about $s_i$ -- whether collisions have happened (and if yes, then with which sample) and which unique atom in $\At(\mcA)$ contains the sample $s_i$.

\begin{definition}[Core-Adaptive Tester]
\label{def:core-adaptive}
    A \emph{core-adaptive tester} for a pair of distributions is an algorithm $T$ that does the following:
    \begin{enumerate}
     \item Fixes $k \in \{1,2\}$ (fixes the distribution $D_1$ or $D_2$ on which to perform the conditional sampling query).
        \item To make $i$-th query, based only on its own internal randomness and the configuration of the previous samples $(c_1,\dots,c_{i-1})$, $T$ provides:
        \begin{enumerate}
            \item \label{step:core-one} A (non-negative) integer $k^A_i$ for each $A \in \At(A_1,\dots, A_{i-1})$ between $0$ and $|A \setminus \{s_1,\dots,s_{i-1}\}|$ (how many fresh -- not already seen -- elements of each particular atom should be included in the next query),
            \item A set $O_i \subseteq \{s_1,\dots, s_{i-1}\}$ (which of the samples $s_1,\dots,s_{i-1}$   will be included in the next query). 
        \end{enumerate}
        \item Based on these specifications, the tester $T$  constructs the $i$-th query set $A_i$  by
        \begin{enumerate}
            \item Drawing uniformly at random, a set $U_i$ from the set
            \begin{align}
                \left\{U \subseteq [n]\setminus \{s_1,\dots,s_{i-1}\} \mid \forall A \in \At(A_1,\dots,A_{i-1}), |U_i \cap A| = k^A_i\right\}
            \end{align}
            i.e., among all the sets containing only ``fresh elements", whose intersection with each atom contains exactly as many elements as $T$ specifies (at Step~\ref{step:core-one}).
            \item $A_i := O_i \cup U_i$.
        \end{enumerate}
       
        \item Samples from $\COND_{D_k}(A_i)$.
        \end{enumerate}
        After $q = q(\epsilon,n)$ queries, the tester $T$ returns ACCEPT or REJECT based  on the configurations  $(c_1,\dots,c_q)$.
\end{definition}
From now on, for brevity, we denote
\[
\At(U_i) : = \left\{U_i \cap A \mid A \in \At(\mathcal{A}_{i-1})\right\}.
\]


It is easy to observe that the conditional sampling queries made by a core-adaptive tester can be viewed in an equivalent way as follows.

\begin{observation}\label{obs:cond-new-defn}
For any distribution $D_k$ ($k \in \{1,2\}$) and $i$-th query set $A_i = O_i \cup U_i$ ($i \in [q]$), the sample obtained from $\COND_{D_k}(A_i)$ can be viewed as:
\begin{enumerate}
    \item First, pick an element $e \in O_i \cup \{U_i\}$ (where $O_i \cup \{U_i\}$ is the set consisting the elements of $O_i$ and the set $U_i$ itself) such that each $j \in O_i$ is picked with probability $\frac{D_k(j)}{D_k(O_i \cup U_i)}$ and $\{U_i\}$ is picked with probability $\frac{D_k(U_i)}{D_k(O_i \cup U_i)}$.    
    \item If $e \in O_i$, then the sample obtained from $\COND_{D_k}(A_i)$ is $e$.
    \item  Otherwise (i.e., if $e = \{U_i\}$), then pick $s' \sim  \COND_{D_k}(U_i)$, and $s'$ is the sample obtained from $\COND_{D_k}(A_i)$. Note that this is equivalent to
    \begin{enumerate}
        \item First picking an atom $V \in \At(U_i)$ with probability $\frac{D_j(V)}{D_j(U_i)}$, and
        \item Then returning a sample $s' \sim  \COND_{D_k}(V)$. 
    \end{enumerate}
\end{enumerate}

\end{observation}

It was shown in~\cite{chakraborty2016power} that general testers are equivalent to core-adaptive testers in terms of testing a label invariant property.

\begin{theorem}[\cite{chakraborty2016power}]
    \label{thm:general-core-adaptive-equiv}
    If there exists a $q$-query general tester for any label invariant property, then there also exists a $q$-query core adaptive tester.
\end{theorem}


A  core adaptive tester $T$, for any $i \le q$,  maps a sequence of configurations $(c_1,\dots,c_{i-1})$ of the received samples so far, via a  function $f_T$, to a pair  $(O_i, U_i)$ (where $O_i \subseteq  \{s_1,\dots,s_{i-1}\}$ and $U_i \subseteq [n]\setminus \{s_1,\dots,s_{i-1}\}$) which determines the $i$-th query set $A_i = O_i \cup U_i$.

In a natural way, a tester $T$ can be fully described by a decision tree $R$. Formally, (the edges of) a path (from root) to any node $v$ at depth $i$ is associated with a sequence of configurations $(c_1,\dots,c_i)$, and the node $v$ is labeled with a pair $(O_v, U_v) = f_T(c_1,\dots,c_{i-1})$ (where $O_v \subseteq \{s_1,\dots,s_{i-1}\}$ and $U_v \subseteq [n] \setminus \{s_1,\dots,s_{i-1}\}$)  which determines the next query set $A_{v} = O_v \cup U_v$. Further, for every possible value of the configuration $c_{i}$ of the sample from $A_{v}$, there is a corresponding child of the node $v$, and the corresponding edge is labeled by the value of the configuration $c_{i}$. Finally, the leaves of $R$ are labeled by either ACCEPT or REJECT. 

We now define a few notations that we will use throughout the paper. We will use $(O_v, U_v)$ for the label of a node $v$ and $A_v = O_v \cup U_v$ for the corresponding query set. Let the nodes of the path from the root to the node $v$ be $v_1,\dots,v_{i}$ where $v_1$ is the root, and $v_i$ is the node $v$. We will use  $\bm{A}_v = (A_{v_1},\dots, A_{v_i})$ for the sequence of the query sets corresponding to the nodes in this path and $\bm{S}_v = (s_{v_1},\dots,s_{v_i})$ to denote the set of samples obtained in the node from the root to $v$. Further, $\At(\bm{A}_v) := \At(A_{v_1},\dots,A_{v_i})$. Note that the set of all the relevant atoms for the tester (described by a decision tree $R$) is  $\cup_{v \in R} \At(\bm{A}_v) $, which we will denote by $\At(R)$.

\subsection{High-level framework}
\label{subsec:highlevelframe}
To demonstrate a lower bound, one of the standard approaches (using Yao's minimax lemma) is, to begin with, two sets of pairs of distributions. The first set, $\YES$ instances, contains pairs of identical distributions. The second set, called $\NO$ instances, has pairs of distributions separated by a (total variation) distance of at least $1/4$ in the total variation distance. The $\YES$ and $\NO$ instances that we consider in this paper (formally defined in Section~\ref{sec:hardinstance}, informally later in this subsection) are slight modifications of the instances considered in~\cite{acharya2018chasm}.\footnote{The change is to simplify our proof, and our proof technique works for their exact instances as well, albeit with certain modifications.} To prove a lower bound on the number of $\COND$ queries of $\Omega(\log \log n)$, it suffices to show that no tester can successfully differentiate between whether an input pair of distributions are drawn from the $\YES$ and $\NO$ instances (with a high probability) unless the tester makes at least $\Omega(\log \log n)$ queries. In other words, we need to show that if the input pair of distributions are drawn from the $\YES$ and $\NO$ instances, the respective distributions on the leaves of the tester are close to each other in the total variation distance if the number of queries made is $q < o(\log \log n)$.

Recall, as mentioned in Section~\ref{sec:core-adapt}, we only focus on proving lower bounds concerning the core-adaptive testers. We need to understand how the path traversed by the tester on the (corresponding) decision tree $R$ depends on the input pair of distributions from $\YES$ and $\NO$ cases. Note the height of $R$ is the number of queries made, i.e., $q$. At any given node $v$, the next node the tester reaches depends on the outcome of the $\COND$ query associated with that node. Suppose we can prove that for any node $v$, the distributions (for $\YES$ and $\NO$ cases respectively) over the outcomes are $\eta$-close in total variation distance. In that case, the overall total variation distance can be at most $q \eta$ since $q$ is the height of the decision tree.
We refer to a node $v$ as \emph{good} if the distributions over the outcomes are close in total variation corresponding to $\YES$ and $\NO$ cases, and \emph{bad} otherwise (more details on good and bad nodes are provided later in this subsection).

As highlighted in Section~\ref{sec:prev}, by~\cite{acharya2018chasm}, the probability that there exists a bad node in $R$ when input is drawn from either $\YES$ or $\NO$ instances  is $\frac{|R|\sqrt{\log n}}{\log n} = o(1)$ only when $|R|  = o(\sqrt{\log n})$. Since we have $|R| = 2^{q^2}$, if  $q = o(\sqrt{\log \log n})$, then there are no bad nodes (with high probability). If there are no bad nodes, indistinguishability follows.

Thus, the bottleneck to show a stronger lower bound will be to reduce $|R|$ (in terms of $q$), which is impossible since the bound is already tight. This is where our novel idea of "decision tree specification" using weaker $\COND$ queries comes into play.

\subsubsection*{Decision Tree Sparsification using $\Ratio$ query}
Let us introduce a new weaker query model called $\Ratio$ query.

\begin{definition}[$\Ratio$ query model]\label{def:ratio-query}
      For any $i$, given the $i$-th query set $A_i = O_i \cup U_i$ and $k \in \{1,2\}$, the weak conditioning query $\Ratio_{D_k}(A_i)$ does the following:
    \begin{enumerate}
        \item \label{step:ratio-first}It picks an element $e$ in $O_i \cup \{U_i\}$ (elements of $O_i$ and the set $U_i$ itself) with probability $\frac{D_j(e)}{D_j(O_i \cup U_i)}$. 
        \item If $e \in O_i$, then the $i$-th sample $s_i = e$. 
        \item Otherwise (i.e., if $e = \{U_i\}$), then
        \begin{enumerate}
            \item \label{step:ratio-second}
            An atom $V \in \At(U_i)$ is picked with probability $\frac{|V|}{|U_i|}$.
            \item \label{step:ratio-third}The sample $s_i$ is generated by the $s_i \sim \COND_{D_k}(V)$. 
    \end{enumerate}
    \end{enumerate}
\end{definition}

Notice the difference between a $\Ratio$ and $\COND$ query is only in Step~\ref{step:ratio-second} (see Observation~\ref{obs:cond-new-defn}). In $\COND$, an atom $V \in \At(U_i)$ is picked with probability $\frac{D_k(V)}{D_k(U_i)}$ whereas in $\Ratio$, it is picked with probability $\frac{|V|}{|U_i|}$ (which is independent of input distributions).

Fix a  decision tree $R$. 
Let $T(\COND)$ and $T(\Ratio)$ denote the tester $T$ that uses the decision tree $R$ (to determine the next node/query) when given access to $\COND$ and $\Ratio$ queries, respectively. We compare the behaviors of these two testers. Before doing so, we would like to emphasize that whether a node is good or bad depends solely on $R$ and the input distributions, it does not depend on oracle access of tester $T$ (we describe the good and bad nodes in the next subsection).
 
For $T(\COND)$, the randomness in picking at atom $V \in \At(U_i)$ (in Step~\ref{step:ratio-second}) is external (i.e., depends on the distribution), whereas, for $T(\Ratio)$, this randomness is internal. Thus a randomized tester with $\Ratio$ oracle can simulate this internal randomness. Any randomized algorithm can be seen as a distribution over its (deterministic) instantiations (obtained by fixing its random choices). Hence, the tester $T(\Ratio)$ can be seen as a distribution over a set of deterministic algorithms such that in each deterministic algorithm, the tester picks a fixed atom $V \in \At(U_i)$ for each query $i \in [q]$. The crucial aspect is that now, the number of possible outcomes (at each node) is at most $|O_i| \le q$ (as compared to the previous $2^q$ that has $\COND$ access, one for each possible atom in $U_i$) resulting in at most $q^q$ nodes (instead of $2^{q^2}$) in the decision tree. Thus, for each decision tree, we can show that if $q = o(\log \log n)$, there are no bad nodes with high probability. Since the original tester is a distribution over these deterministic counterparts, it is easy to argue that the $T(\Ratio)$ does not pass through a bad node either with high probability if $q = o(\log \log n)$. Note that this does not contradict what we said earlier that the previous analysis is tight --- we show that ``the tester does not pass through a bad node with high probability" as opposed to ``there are no bad nodes". 

However, so far, we argue the results for $T(\Ratio)$, but ultimately we want it to hold for $T(\COND)$. Surprisingly we show that the distributions of the run of a tester are close in total variation distance for both cases. The following theorem is the heart of our proof.
 \begin{lemma}[Informal Statement of Lemma~\ref{lem:good1-high-prob}]
     \label{thm:main1}
     For the instances we consider, the distributions on leaves of the decision tree $R$ are close in total variation distance, for the cases when the tester is provided access to $\COND$ queries and $\Ratio$ queries respectively, unless $q = \Tilde{\Omega}(\log \log n)$. 
 \end{lemma}
 
The above theorem is pivotal because, by only an additional $o(1)$ probability, we can show that even $T(\COND)$ does not pass through a bad node. Finally, in Lemma~\ref{lem:if-good-then-lb}, we show that conditioned on the event that the tester does not pass through a bad node, the $\YES$ and $\NO$ instances are indistinguishable unless  $q = \Tilde{\Omega}(\log \log n)$.

\begin{lemma}[Informal Statement of Lemma~\ref{lem:if-good-then-lb}]
     \label{thm:main}
     The distributions on leaves of the decision tree $R$ are close in total variation distance for $\YES$ and $\NO$ instances in $T(\COND)$  unless $q = \Tilde{\Omega}(\log \log n)$. 
 \end{lemma}

\subsubsection*{On our $\YES$ and $\NO$ instances}
We informally explain our $\YES$ and $\NO$ instances (formally defined in Section~\ref{sec:hardinstance}). For $\YES$ instances, we consider a pair of identical distributions $(Q_1, Q_1)$. On the other hand, for $\NO$ instances, we consider a pair of $1/4$-far distributions $(Q_1, Q_2)$. For both $Q_1$ and $Q_2$, the support of the distributions is randomly partitioned into $\Theta(\sqrt{\log n})$ buckets, with each successive bucket increasing geometrically (more specifically, by a factor of exponential in $\sqrt{\log n}$) in size; however, the probability mass of each bucket remains the same (Figure~\ref{fig:buckets}). In other words, if we compare the probability of two elements from two different buckets, they differ (slightly) sub-polynomially (in $n$). We leverage this property crucially in our lower-bound argument. The distributions $Q_1$ and $Q_2$ differ within each bucket. Inside any bucket, $Q_1$ is uniform on its elements. In contrast, in $Q_2$, a random half of elements have probability a  constant factor times the other random half of elements (see Figure~\ref{fig:inside-buckets}). Let us refer to these two random partitions of buckets as sub-buckets.
 
 It is important to note that if we could get two random samples, say $s$ and $s'$, from any particular bucket, it would be easy to distinguish between the two cases. The reason is that $s$ and $s'$ would be in two different sub-buckets with constant probability. Then for $Q_1$,  we have $Q_1(s) = Q_1(s')$ whereas for $Q_2$, we have $Q_2(s)/Q_2(s') = \Omega(1)$ so a $\COND(s,s')$ query would detect whether the input pair of distributions is $(Q_1,Q_1)$ or $(Q_1,Q_2)$.  Indeed, if the support size is known to the tester, then in $O(1)$-query, it is possible to pick two random samples from the same bucket. That is why our instances' support size is also set randomly.

\subsubsection*{Good and Bad nodes}
  Now we explain the good and bad nodes in detail. Ideally, we want to say a node is \emph{good} if the outcomes are close in the total variation distance for $\YES$ and $\NO$ instances. It happens if three conditions are satisfied (see Definition~\ref{def:good node}), which we summarize below.

 In any atom $A \in \At(R)$, the expected number of items from the $j$-th bucket is $\frac{|A|b \rho^j}{n}$. We would like that the actual number is concentrated tightly around this expectation for all atoms and all buckets. By an application of standard concentration inequalities, this would happen if this expected value is either too large or too small compared to 1 (see Lemma~\ref{lem:phi}). Next, we also want that an atom to either intersects a large number of buckets  (at least $\poly(\log \log n)$) -- such atoms are called \emph{large} atoms -- or not intersect with any bucket (such atoms are called \emph{small} atoms). This condition would be pivotal to argue that with high probability, no two samples can be from the same bucket (first item of Definition~\ref{def:good2}). Finally, we want the probability mass of $U_v$ (at a node $v$) to be too heavy or too light compared to the probability mass of any single sampled element (second item of Definition~\ref{def:good2}). If none of the above three conditions holds, we call such a node a \emph{bad} node.

It is worth noting that our definition of a node being good or bad is (almost) equivalent to that of ACK~\cite{acharya2018chasm}. Our main ingenuity lies in (i) considering the event ``the tester does not pass through a bad node with high probability" as opposed to ``there are no bad nodes", and (ii) introducing the weaker query model $\Ratio$ (and showing equivalence to $\COND$) to analyze the event ``the tester does not pass through a bad node with high probability".

\subsubsection*{Indistinguishability between $\YES$ and $\NO$}
Assuming that the tester does not pass through a bad node, we show that the input pair of distributions from the $\YES$ and $\NO$ instances are indistinguishable. As we argued before, we need to show that the induced distributions on the leaves of the decision tree are close in the total variation distance for both cases.

The main challenge is that even if we assume the testers do not pass through a bad node, the decision to choose the next node by the tester depends not only on the previous traversed nodes but also on the bucket distribution of the seen samples (that is, which buckets the seen samples belong to). Note that the bucket distributions are not part of configurations. That is why we argue that the bucket distributions, as well as the distributions of the next configuration for both $\YES$ and $\NO$ cases, are close in total variation distance at every step. We formally argue this in Section~\ref{sec:ifGoodthenlb}.

\subsection{Applying our technique to testing label-invariant property}
Firstly, note that our lower bound of  $\Omega(\log \log n)$ for equivalence 
testing will not directly give a lower bound on a label invariant property. In
the equivalence testing problem, the input is a pair of distributions, while 
the testing of a label invariant property has only one distribution as input.
One may fix one of the distributions in the equivalence testing and regard the
other distribution as the input. Interestingly, in this case, the number of
required queries becomes independent of $n$~\cite{falahatgar2015faster}. 

To obtain a lower bound on testing a label invariant property, we proceed along the path already charted by CFGM~\cite{chakraborty2016power}. They defined a property called \emph{even uniblock property} and used it to prove the lower bound of $\sqrt{\log\log n}$. In fact, for our lower bounds, we even use the same hard instances as defined in \cite{chakraborty2016power}. The crucial point where we do better than the previous attempts is that, like in the equivalence testing problem, our ``$\Bad$ event" is when a run of the algorithm passes through a ``$\Bad$ node" compared to the previous attempts where the ``$\Bad$ event" was defined as when there is some ``$\Bad$ node" in the decision tree. Of course, this crucial difference is significant as we can show via the $\Ratio$ model that if the number of queries is $\Omega(\log\log n)$, the $\Bad$ event happens with very low probability. Our proof, as in the case of the lower bound of equivalence testing, has three parts: (a) showing that a run of the $\Ratio$ model algorithm does not pass through a bad node with high probability if $q = o(\log \log n)$, (b) showing that the total variation distance between the distribution on leaves of the decision tree for the $\Ratio$ model and the $\COND$ model is small, (c) assuming that a run of the algorithm does not pass through a bad node, the total variation distance over leaves is small for $\YES$ and $\NO$ cases.


\subsection{Corollary to the support size estimation}
\label{sec:support}
We show that if we can know the support size of the given distribution, then in $O(1)$ queries, we can distinguish between our $\YES$ and $\NO$ instances constructed for the equivalence testing problem (as we mentioned before, it is for this reason that we randomly set the support size in our instances). Therefore, the lower bound for the equivalence testing in Theorem~\ref{thm:mainlb} also holds for the support size estimation.

We quickly recap the informal construction of our $\YES$ and $\NO$ instances (see section~\ref{sec:hardinstance} for the formal definition). For $\YES$ instances, we consider a pair of identical distributions $(Q_1, Q_1)$. On the other hand, for $\NO$ instances, we consider a pair of $1/4$-far distributions $(Q_1, Q_2)$.  For both $Q_1$ and $Q_2$, the support of the distributions are the same and randomly partitioned into $\Theta(\sqrt{\log n})$ buckets, with each successive bucket increasing geometrically (by a factor of exponential in $\sqrt{\log n}$, to be precise) in size.  The distributions $Q_1$ and $Q_2$ differ within each bucket. Inside any bucket, $Q_1$ is uniform on its elements. In contrast, in $Q_2$, a random half of elements have probability a  constant factor times the other random half of elements (see Figure~\ref{fig:inside-buckets}). These two random partitions of buckets are called sub-buckets.

Let $(D_1, D_2)$ be the pair of distributions provided to the tester (which could be either $(Q_1, Q_1)$ or $(Q_1, Q_2)$). Note that by our construction, distributions $D_1$ and $D_2$ have the same support size (say $s$). Assume that the support size $s$ is known. Let $c$ be a large constant.
  We construct a set $S$  by including each element $i \in [n]$ with probability $c/s$. 
  Thus on expectation, there will be $c$ elements from the support in $S$.
  Note that the size of each successive bucket differs by a factor of $2^{\sqrt{\log n}}$. Hence on expectation, there will be $O(1) (\approx c/2)$ elements from both the sub-buckets of the largest bucket (as both the sub-buckets are of equal size) and $o(1)$ elements from the rest of the buckets. Therefore, with a high probability, there will be $O(1)$ elements from both the sub-buckets of the largest bucket and no elements from the rest of the buckets. 
    Since there are only $O(1)$ elements from the support in $S$, it is possible to identify these elements using only $O(1)$ queries. Finally,  for each pair $s,s'$ among these elements,  we perform $\COND_{D_2}\left(\{s,s'\}\right)$ queries.  We note that in the $\YES$ instance, for each pair $s,s'$, $\COND_{D_2}\left(\{s,s'\}\right)$ is a uniform distribution on $\{s,s'\}$ whereas in the $\NO$ instance, for at least one pair $s,s'$ (when $s$ and $s'$ belongs to different sub-buckets),  $\COND_{D_2}(s,s')$ assigns probability $3/4$  to one and $1/4$ to another. Hence, using overall $O(1)$ queries (assuming the support size is known), one can distinguish the $\YES$ and $\NO$ instances.


\begin{corollary}
  Any (randomized) adaptive tester for estimating the support size of a distribution to a $4/3$-multiplicative factor must make  $\Tilde{\Omega}(\log \log n)$ $\COND$ queries even when it is promised that the probability of any element in the support is at least $1/n$.
\end{corollary}



\section{$\Tilde{\Omega}(\log \log n)$ lower bound for Equivalence testing}\label{sec:main}

In this section, we prove Theorem~\ref{thm:mainlb}. We will eventually use Theorem~\ref{thm:extyao} to prove our theorem. For that purpose, we start with defining the hard-to-distinguish $\YES$ and $\NO$ instances in Section~\ref{sec:hardinstance}. Then we define the $\Good$ event in Section~\ref{sec:goodevent}. The proof of Theorem~\ref{thm:mainlb} follows from two crucial lemmas - namely that the $\Good$ event happens with high probability and that if $\Good$ event happens, then the $\YES$ and $\NO$ instances are hard to distinguish. Lemma~\ref{lem:Bad} proves that the $\Good$ event happens with high probability, and this is proved in Section~\ref{sec:bad}. Lemma~\ref{lem:if-good-then-lb} states that if the $\Good$ event happens, then $\YES$ and $\NO$ instances are indistinguishable, and this is proved in Section~\ref{sec:ifGoodthenlb}.

\subsection{Distributions over the YES and NO instances}\label{sec:hardinstance}
Let us start by describing a distribution $\mathcal{I}_{\YES}$ over the $\YES$-instances and a distribution $\mathcal{I}_{\NO}$  over the $\NO$-instances. We describe them by describing the process using which an element of the $\YES$-instances (and $\NO$-instances respectively) is produced.

We will now present a randomized procedure to generate a pair of distributions $(Q_1, Q_2)$. The $\YES$ instance will have both distributions as $Q_1$, whereas the $\NO$ instance will have the two distributions as $Q_1$ and $Q_2$. 

\begin{enumerate}
    \item An integer $\kappa \in \left\{0,1,\dots,\left\lfloor \frac{\log n}{2} \right\rfloor \right\}$ is chosen uniformly at random. We set $b=2^{\kappa}$, $\rho = 2^{\sqrt{\log n}}$, $\tau = \frac{\sqrt{\log n}}{4}$, and let $m = b(\rho+\rho^2+\dots+\rho^{\tau})$.
    \item A pair of distributions $(Q_1, Q_2)$ is constructed as follows (see Figure~\ref{fig:buckets} and \ref{fig:inside-buckets}):
    \begin{itemize}
        \item We randomly choose a subset of size $m$ from $[n]$ that forms the support of both distributions $Q_1$ and $Q_2$.
        \item We then randomly partition the support into $\tau$ buckets (for both $Q_1$ and $Q_2$) $B_1,\dots,B_{\tau}$ such that $|B_j| = b \rho^j$ and assign the probability $1/\tau$ to each bucket.
        \item In distribution $Q_1$, the probability distribution is uniform in each bucket, i.e., for all $j \in [\tau]$, and $i \in B_j$, we have  $Q_1(i) = \frac{1}{\tau b \rho^j}$.
        \item In distribution $Q_2$, each bucket $B_j$ is randomly partitioned into two subbuckets $B^h_j$ and  $B^{\ell}_j$ of equal size ($ = \frac{|B_j|}{2}$) such that $Q_2(B^h_j) = \frac{3}{4\tau}$ and $Q_2(B^{\ell}_j) = \frac{1}{4\tau}$ and distribution is uniform in each sub-buckets. In other words, for all $j \in [\tau]$, we have
        \[
        Q_2(i) = \begin{cases}
        \frac{3}{2\tau b \rho^j} \quad \text{if }i \in B^h_j\\
        \frac{1}{2\tau b \rho^j} \quad \text{if }i \in B^{\ell}_j
        \end{cases}
        \]
    \end{itemize}
\end{enumerate}

\begin{figure}
    \centering
    \begin{tikzpicture}[scale=0.4]
    \node[circle,inner sep=0pt,minimum size=3pt,fill=black] (a) at (0,0) {};
    \node[circle,inner sep=0pt,minimum size=3pt,fill=black] (b) at (1,0) {};
    \node[circle,inner sep=0pt,minimum size=3pt,fill=black] (c) at (3,0) {};
    \node[circle,inner sep=0pt,minimum size=3pt,fill=black] (d) at (7,0) {};
    \node[circle,inner sep=0pt,minimum size=3pt,fill=black] (e) at (15,0) {};
    \node[circle,inner sep=0pt,minimum size=3pt,fill=black] (f) at (31,0) {};
    
    \draw[] (a) -- (b) -- (c) -- (d) --(e) -- (f);

\node[] (l) at (0.5,2) {$B_1$};
\node[] (m) at (2,2) {$B_2$};
  \node[] (n) at (23,2) {$B_\tau$}; 
    \end{tikzpicture}
 \caption{For both $Q_1$ and $Q_2$, $|B_j| = b \rho^j$ and $Q_1(B_j) = Q_2(B_j) = \frac{1}{\tau}$}  
 \label{fig:buckets}
 \vspace{0.4cm}
 \begin{tikzpicture}[scale = 0.4]
      \node[circle,inner sep=0pt,minimum size=3pt,fill=black] (a) at (0,0) {};
    \node[circle,inner sep=0pt,minimum size=3pt,fill=black] (b) at (10,0) {};
     \node[circle,inner sep=0pt,minimum size=3pt,fill=black] (c) at (20,0) {};
     \draw[] (a) -- (b) -- (c);
     \node[] (l) at (5,6) {$B^l_j$};
\node[] (m) at (15,6) {$B^h_j$};

\node[circle,inner sep=0pt,minimum size=3pt,fill=black] (a') at (0,-3) {};
    \node[circle,inner sep=0pt,minimum size=3pt,fill=black] (b') at (10,-3) {};
 \draw[red] (a') -- (b');

\node[circle,inner sep=0pt,minimum size=3pt,fill=black] (b'') at (10,3) {};
    \node[circle,inner sep=0pt,minimum size=3pt,fill=black] (c'') at (20,3) {};
 \draw[red] (b'') -- (c'');

 \draw[dotted] (b'') -- (b');

\node[] () at (-1,0) {$Q_1$};
\node[] () at (-1,-3) {$Q_2$};
 
 \end{tikzpicture}
 \caption{$B^l_j$ and $B^h_j$ are random partition of $B_j$ such that $|B^l_j| = |B^h_j|$. We have $Q_1(i) = \frac{1}{\tau b \rho^j}$ for all $i \in B_j$ whereas $Q_2(i) = \frac{3}{2 \tau b \rho^j}$ for $ i \in B^h_j$ and  $Q_2(i) = \frac{1}{2 \tau b \rho^j}$ for $ i \in B^l_j$. }
    \label{fig:inside-buckets}
\end{figure}

By the construction of the NO instance $(Q_1,Q_2)$, it follows immediately that

\begin{observation}
$\dtv (Q_1, Q_2)= \frac{1}{2} \sum_{i \in [\tau]} \frac{1}{2\tau b \rho^j} \cdot b \rho^j = \frac{1}{4}$.
\end{observation}


\subsection{The $\Good$ Event}\label{sec:goodevent}

For any adaptive algorithm $\Alg$ and input $x$ (which is a pair of distributions $D_1, D_2$; recall, from the previous section that for YES instance $D_1=D_2=Q_1$, and for NO instance $D_1=Q_1$ and $D_2=Q_2$), let us first define the $\Good$ event that we will use with Theorem~\ref{thm:extyao} to prove Theorem~\ref{thm:mainlb}.

We will assume (from Theorem~\ref{thm:general-core-adaptive-equiv}) the algorithm $\Alg$ is a core-adaptive algorithm. 
We will define events $\Good_1(\Alg, x)$, $\Good_2(\Alg, x)$ and $\Good_3(\Alg, x)$ and finally we will have an event $\Good(\Alg, x) = \Good_1(\Alg, x) \wedge \Good_2(\Alg, x) \wedge \Good_3(\Alg, x)$. 


Throughout the rest of the paper, we will assume
\[
\gamma = (\log \log n)^{9},\quad \quad \alpha = (\log n)^{(\log \log n)^2},\quad \quad \phi = (\log \log n)^{20}.
\]
Further, recall that as mentioned in Section~\ref{sec:hardinstance},
\[
\rho = 2^{\sqrt{\log n}},\quad \quad \tau = \frac{\sqrt{\log n}}{4},\quad \quad b=2^{\kappa}
\]
where  $\kappa$ is chosen uniformly at random from $\left\{0,1,\dots,\left\lfloor \frac{\log n}{2} \right\rfloor \right\}$ .

\paragraph{The $\Good_1(\Alg, x)$ event: }
We start with defining $\phi_A$ associated with each atom $A$. 

\begin{definition}
\label{def:phi}
    For any atom $A$, $\phi_A$ is the smallest  integer $\triangle \in \{0,1,\dots,\tau-1\}$ such that  $\frac{|A|b \rho^{\tau-\triangle}}{n} < 1/\alpha$ and if such an integer does not exist then $\phi_A =\tau$. 
\end{definition}
Note that the expected size of the intersection of atom $A$ with bucket $B_{\tau - \triangle}$ is $\frac{|A|b\rho^{\tau - \triangle}}{n}$. Thus, $\phi_A$ is the number of buckets with which $A$ has a large intersection in expectation.

\begin{definition}[Good and Bad node]
\label{def:good node}
    For a node $v$ of the decision tree $R$, recall that $\bm{A}_v = (A_{v_0},\dots, A_{v_i})$ is the set of all queries made in the path from the root to the node $v$ in the decision tree. The node $v$ is called \emph{good} for $x$ if it satisfies all of the following conditions.
    \begin{enumerate}
        \item For any atom $A \in \At(\bm{A}_v)$, we have for every bucket $j \in [\tau]$, either $ \frac{|A|b\rho^j}{n} \ge \alpha$ or $\frac{|A|b\rho^j}{n} \le \frac{1}{\alpha}$. 
        \item For any atom $A \in \At(\bm{A}_v)$, we have either $\frac{|A|b \rho^{\tau-\phi}}{n} \ge \alpha$ or $\frac{|A|b \rho^{\tau}}{n} \le 1/\alpha$. If the former condition holds, we say atom $A$ is large, and if the latter condition holds, we say $A$ is small,
        \item For all $U_{v_\ell}$ ($\ell \le i$), we have for all $j \in [\tau]$, either $\frac{\sum_{A \in \At(U_{v_\ell})} \phi_A |A|}{\tau n} \ge  \frac{\gamma}{\tau b\rho^j}$ or $\frac{\sum_{A \in \At(U_{v_\ell})} \phi_A |A|}{\tau n} \le  \frac{1}{\gamma \tau b\rho^j}$.
    \end{enumerate}
     A node $v$ is called \emph{bad} for $x$ if it is not good for $x$.
\end{definition}
In simple terms, if the algorithm $\Alg$  is currently at a good node $v$, then so far (1) for all atoms, the expected intersection size with all buckets is either large (at least $\alpha$) or small (at most $1/\alpha$). This will help in showing tight concentration bounds (Lemma~\ref{lem:phi}) as, by Chernoff bound, the actual value is tightly concentrated to its expected value if the expectation is large or small, 
(2) any atom either has a large intersection with at least $\phi$ buckets (such atoms are large atoms)  or a small intersection with all buckets (such atoms are small atoms), (3)  the expected probability mass of  any unseen query set   is `incomparable' to the probability of any element in the entire distribution, that is, their ratio is either at least $\gamma$ or at most $1/\gamma$.
\begin{definition}
    The event $\Good_1(\Alg, x)$ is defined as ``the run of the algorithm $\Alg$ on input $x$ does not pass through a bad node."
\end{definition}

Before we give the definition of the events $\Good_2$ and $\Good_3$,  we prove concentration bounds for the good nodes.

\begin{lemma}\label{lem:phi}
    
    With probability at least $1 - \frac{1}{\poly(\log \log n)}$, for all good nodes $v$ and for all atoms $A \in \At(\bm{A}_v)$ , 
\begin{align}
    & \mbox{if $j \le \tau - \phi_A$ then $|B_j \cap A| = 0$, and,} \label{eq:goodnode1}\\
    & \mbox{if $j \ge \tau - \phi_A +1$ then } |A \cap B_j| \in \left[1-\frac{1}{\gamma}, 1+ \frac{1}{\gamma}\right] \frac{|A|b \rho^j}{n} \label{eq:goodnode2}.
\end{align}
    
\end{lemma}

\begin{proof}
  
    If  $\frac{|A|b\rho^j}{n} \le \frac{1}{\alpha}$ for any atom $A$ and $j \in [\tau]$ then by Markov's inequality $\Pr[|A \cap B_j| > 0] \le \frac{1}{\alpha}$.  Taking union bound over all atoms corresponding to all good nodes (which can be at most $2^{q+ q^2} \le (\log n)^{\log \log n}$) and all $j \in [\tau]$, with probability at least $1 - \frac{1}{\poly(\log \log n)}$, 
    \begin{equation}
        \mbox{ for any atom $A$ and $j \le \tau - \phi_A$, we have $|B_j \cap A| = 0$}.
    \end{equation}

    Now note that if $j \ge \tau - \phi_A +1$ then
    $\E\left[|B^h_j \cap A|\right] = \frac{|A|b \rho^j}{2n} \ge \alpha/2$. Therefore, by Lemma~\ref{lem:chernoff}, 
    $$\Pr\left[|A \cap B^h_j| \not \in \left[1-\frac{1}{\gamma}, 1+ \frac{1}{\gamma}\right] \frac{|A|b \rho^j}{2n}\right] \le O(exp(-\alpha/\gamma^2)).$$
    Similarly, $$\Pr\left[|A \cap B^{\ell}_j| \not \in \left[1-\frac{1}{\gamma}, 1+ \frac{1}{\gamma}\right] \frac{|A|b \rho^j}{2n}\right] \le O(exp(-\alpha/\gamma^2)).$$ Hence 
    \begin{equation}
    \Pr\left[|A \cap B_j| \not \in \left[1-\frac{1}{\gamma}, 1+ \frac{1}{\gamma}\right] \frac{|A|b \rho^j}{n}\right] \le O(exp(-\alpha/\gamma^2)).\end{equation}
\end{proof}

From Lemma~\ref{lem:phi}, we deduce the following useful lemma.

\begin{lemma}
\label{cor:prob-mass-conc}
   With  probability at least $1 - o(1)$, for all good nodes $v$ and for all  atoms $A \in \At(\bm{A}_v)$, 
   \begin{align}
    & \mbox{if $j \le \tau - \phi_A$ then $D_1(A \cap B_j)= D_2(A \cap B_j) = 0$, and,} \\
    & \mbox{if $j \ge \tau - \phi_A +1$ then } D_1(A \cap B_j),D_2(A \cap B_j) \in \left[1-\frac{1}{\gamma}, 1+ \frac{1}{\gamma}\right] \frac{|A|}{\tau n}.
\end{align}
   Hence, $D_1(A),D_2(A) \in [1-1/\gamma,1+1/\gamma] \frac{\phi_A |A|}{\tau n}$. Therefore, for any query set $U_{v_\ell}$ ($\ell \le i$), we have $$D_1(U_{v_\ell}),D_2(U_{v_\ell}) \in \left[1-\frac{1}{\gamma},1+\frac{1}{\gamma}\right] \frac{\sum_{A \in \At(U_{v_\ell})} \phi_A |A|}{\tau n}.$$
\end{lemma}

\begin{proof} 
From Equation~\ref{eq:goodnode1} and \ref{eq:goodnode2} in Lemma~\ref{lem:phi} we have 
    $$D_k(A) = \sum_{j} \frac{|A \cap B_j|}{b\tau \rho^j} \in  \left[1-\frac{1}{\gamma}, 1+ \frac{1}{\gamma}\right] \sum_{j:\ge r - \phi_A +1} \frac{|A \cap B_j|}{b\tau \rho^j} = \left[1-\frac{1}{\gamma}, 1+ \frac{1}{\gamma}\right]\frac{\phi_A |A|}{\tau n}.$$
\end{proof}

\paragraph{The $\Good_2(\Alg, x)$ event: }
The input $x$ (we are concerned about) is drawn according to the $\IN$ or $\IY$.
Let  $b(s_i)$ denote the index of the bucket to which the seen element $s_i$ belongs to. Recall the notation $\bm{S}_i = (s_1,\dots,s_i)$. We will use $\bm{B}(\bm{S}_{i})$ for $(b(s_1),\dots,b(s_i))$.
\begin{definition}\label{def:good2}
Let $\Good_2(\Alg, x)$ be the event  defined as:
\begin{enumerate}
    \item   For all $i \in [q]$,  for any $s,s' \in \bm{S}_i$ with $D_k(s),D_k(s') \ne 0$ for any $k \in \{1,2\}$, we have $b(s) \neq b(s')$.
    \item  For all $i \in [q]$,   the ratio  $\frac{D_k(U_{i})}{D_k(s)}$ for any $k \in \{1,2\}$ and $s \in \bm{S}_i$ is not in $[\frac{1}{3\gamma},3\gamma]$.
\end{enumerate}
\end{definition}

\paragraph{The $\Good_3(\Alg, x)$ event: }  An $e_h \in O_i \cup \{U_i\}$ is called the unique heaviest element in $ O_i \cup \{U_i\}$ if $D_k(e_h) > D_{k}(e)$ for all $k \in \{1,2\}$ and $e \in  O_i \cup \{U_i\}$ such that $e \neq e_h$.
Let $\Good_3$ be the event that for all $i \in [q]$, the outcome of $\COND_{D_k}(O_i \cup \{U_i\})$ is the unique heaviest element in $O_i \cup \{U_i\}$ for all $i \in [q]$.

\paragraph{The $\Good(\Alg, x)$ event: }


We define the event  
\begin{equation}\label{def:good}
    \Good  = \Good_1 \wedge \Good_2 \wedge \Good_3.
\end{equation}

\subsection{Proof of Theorem~\ref{thm:mainlb}}

We will use Theorem~\ref{thm:extyao} to prove Theorem~\ref{thm:mainlb}. The event $\Bad(\Alg, x)$ is the event $\overline{\Good(\Alg, x)}$, where $\Good(\Alg, x)$ is the event defined in Definition~\ref{def:good}. So we note that Theorem~\ref{thm:mainlb} follows from the following two lemmas (Lemma~\ref{lem:Bad} and \ref{lem:if-good-then-lb}).

\begin{lemma}\label{lem:Bad} For any deterministic core-adapative tester $\Alg$ that makes $q \le \frac{\log \log n}{100 \log \log \log n}$  queries,
     \begin{enumerate}
         \item $\Pr_{x\in_R \IY}[\Good(\Alg, x)] \geq 1 - \delta/2$,
        \item $\Pr_{x\in_R \IN}[\Good(\Alg, x)] \geq 1- \delta/2$.
     \end{enumerate}
\end{lemma}

\begin{lemma}\label{lem:if-good-then-lb} For any deterministic core-adapative tester $\Alg$ that makes $q \le \frac{\log \log n}{100 \log \log \log n}$  queries,
for any $\sigma \in [n]^q$ 
    \begin{align*} 
\Pr_{x\in_R \IY}[\mbox{ the answers to the $q$} & \mbox{ queries made by $\Alg(x)$ is $\sigma$ } \mid \Good(\Alg, x)] \\
\leq & \frac{3}{2}\Pr_{x\in_R \IN}[\mbox{ the answers to the $q$ queries made by $\Alg(x)$ is $\sigma$ } \mid \Good(\Alg, x)].
\end{align*}
\end{lemma}

We will prove Lemma~\ref{lem:if-good-then-lb} in Section~\ref{sec:ifGoodthenlb} and the proof of Lemma~\ref{lem:Bad} in Section~\ref{sec:bad}.

\subsection{Proof of Lemma~\ref{lem:Bad}}
\label{sec:bad}

Each of the parts of Lemma~\ref{lem:Bad} can be proved in three steps. First, we will prove that $\Good_1$ happens with probability at least $(1-\delta/6)$. Then we show that assuming $\Good_1$, the event $\Good_2$ happens with probability at least $(1-\delta/6)$. and finally assuming $\Good_1$ and $\Good_2$ then the event $\Good_3$ will happen with probability at least $(1-\delta/6)$. 

The lower bound on the probability of $\Good_1$ is a main technical lemma, and we present the precise statement given in Lemma~\ref{lem:good1-high-prob}. The proof of Lemma~\ref{lem:good1-high-prob} is presented in Section~\ref{sec:good1}.

\begin{lemma}
\label{lem:good1-high-prob}
    If the number of adaptive queries made by the algorithm $\Alg$ is  $q \le \frac{\log \log n}{100 \log \log \log n}$,  then 
    \begin{enumerate}
        \item $\Pr_{x\in_R \IY} \left[\Good_1(\Alg, x)\right] \geq 1-\frac{\delta}{6}$,
        \item $\Pr_{x\in_R \IN} \left[\Good_1(\Alg, x)\right] \geq 1-\frac{\delta}{6}$.
    \end{enumerate}
\end{lemma}




We will now prove that conditioned on  the event $\Good_1$ holds the event $\Good_2$ holds with high probability. Let us consider the case where $x$ is drawn according to $\IN$. We note that the equivalent statement holds for the case when $x$ is drawn according to $\IY$. 

We will bound the probability that $s_i$ falls in the same bucket corresponding to some previous samples. The seen 
element (with non-zero probability) can 
only come from the query $\COND_{D_k}(W)$ for 
some large atom $W$ for any $k \in \{1,2\}$. Note that for any large atom 
$W$, we have for all $k \in \{1,2\}$, 
$D_k(W \cap B_j) \in
[1-\frac{1}{\gamma},1+\frac{1}{\gamma}]\frac{|W|}{\tau n}$ for 
all $j \ge r - \phi_W +1$ and $D_j(W \cap B_j) = 0$ for $j \le r- \phi_W$ . 
Therefore, 
$$\Pr\left[\mbox{ the 
sample from $\COND_{D_k}(W)$ is in  bucket $B_j$}\right] = \left\{
	\begin{array}{ll}
		\in [1-1/\gamma,1+1/\gamma] \frac{1}{\phi_W}  & \mbox{if } j \ge r-\phi_W+1 \\
		= 0 & \mbox{if } j \le r - \phi_W
	\end{array}
\right.$$


Consider the inductive hypothesis - for all  $i \in [q]$, with probability at least $1-\frac{i^2}{\phi}$, we have $b(s) \neq b(s')$ for all $s,s' \in \bm{S}_{i}$.  Assuming the inductive hypothesis true for $i-1$, the probability that the $i$-th sample $s_i$ falls into any   bucket in $\bm{B}(\bm{S}_{i-1})$ is at most $(1+1/\gamma) \frac{i-1}{\phi_W} \le (1+1/\gamma) \frac{i-1}{\phi} \le \frac{2(i-1)}{\phi}$. In other words, with probability at least $1 - \frac{(i-1)^2}{\phi} - \frac{2(i-1)}{\phi} \ge 1 - \frac{i^2}{\phi},$   we have $b(s) \neq b(s')$ for all $s,s' \in \bm{S}_{i}$,   i.e., the induction hypothesis is true for $i$ as well. Hence the first point of event $\Good_2$ happens with probability at least $1-q^2/\phi$. 

Now we show that the second point happens with high probability assuming that the first point happens.
Assuming $\Good_1$ happens, for all $U_{i}$, we have for all $j \in [\tau]$,
 \begin{itemize}
     \item Either $\frac{\sum_{A \in \At(U_{i})} \phi_A |A|}{\tau n} \ge  \frac{\gamma}{\tau b\rho^j} \ge \frac{2 D_k(s_i)}{3\gamma}$,
     \item Or $\frac{\sum_{A \in \At(U_{i})} \phi_A |A|}{\tau n} \le  \frac{1}{\gamma \tau b\rho^j} \leq 2\gamma D_k(s_i)$.
 \end{itemize}

 Further by Lemma~\ref{cor:prob-mass-conc}, we have 
 $$D_k(U_i) \in \left[1-\frac{1}{\gamma},1+\frac{1}{\gamma}\right] \sum_{A \in \At(U_i)} \frac{\phi_A |A|}{\tau n}.$$ 

Therefore, \begin{equation}\label{eq:good3-1}
    \mbox{$D_k(U_i) \leq (1 + \frac{1}{\gamma})\cdot 2\gamma D_k(s_i) \leq 3\gamma D_k(s_i)$ or $D_k(U_i) \geq (1 - \frac{1}{\gamma})\cdot \frac{2D_k(s_i)}{3\gamma}\geq \frac{D_k(s_i)}{3\gamma}$.}
\end{equation}

Thus, we have 
\begin{equation}\label{eq:good2}
\Pr_{x\in_R \IN} \left[\Good_2(\Alg, x) \mid \Good_1(\Alg, x)\right] \geq 1-\frac{\delta}{6} \mbox{ and }
        \Pr_{x\in_R \IY} \left[\Good_2(\Alg, x) \mid \Good_1(\Alg, x)\right]\geq 1-\frac{\delta}{6}.\end{equation}



Now let us assume that the event $\Good_2$ happens. 
Note that $\rho >> \gamma$ and thus if $b(s) \neq b(s')$, we have
\begin{equation}\label{eq:good3-2}
    \mbox{either $\frac{D_k(s)}{D_k(s')} \ge \rho \ge \gamma$ or $\frac{D_k(s)}{D_k(s')} \le 1/\rho\le 1/\gamma$.} 
\end{equation}

Let $e_h$ be the unique heaviest element in $O_i\cup \{U_i\}$. From Equation~\ref{eq:good3-1} and \ref{eq:good3-2} we note that $$\sum_{e \in (O_i \cup \{U_i\})\setminus e_h}D_k(e) \le \frac{D_k(e_h)}{\gamma} + \frac{D_k(e_h)}{\gamma^2} + \cdots \le 2 D_k(e_h)/\gamma.$$   Therefore, with probability at least $(1 - 2/\gamma)$, the outcome of the conditional query $\COND_{D_k}(O_i \cup \{U_i\})$ is  $e_h \in O_i \cup \{U_i\}$. Thus,

   \begin{equation}\label{eq:good3}
       \Pr_{x\in_R \IN} \left[\Good_3(\Alg, x) \mid \Good_1(\Alg, x)\wedge \Good_2(\Alg, x)\right], 
       \Pr_{x\in_R \IY} \left[\Good_3(\Alg, x) \mid \Good_1(\Alg, x)\wedge \Good_2(\Alg, x)\right]
       \geq 1-\frac{\delta}{6} .
   \end{equation}

From Lemma~\ref{lem:good1-high-prob} and Equation~\ref{eq:good2} and \ref{eq:good3} we have Lemma~\ref{lem:Bad}.

\subsubsection{Proof of Lemma~\ref{lem:good1-high-prob}}\label{sec:good1}
In this subsection, we prove Lemma~\ref{lem:good1-high-prob}, i.e., $\Good_1(\Alg,x)$ happens with high probability (when $\Alg$ has an access to a $\COND$ query model). 



Informally speaking, we first argue that if $\Alg$ had access to a $\Ratio$ oracle (instead of $\COND$), then $\Good_1(\Alg,x)$ would happen with high probability. Let us consider the event $\G_1$ that denotes the run of an algorithm $\Alg$ having access to $\Ratio$ oracle, on input $x$ does not pass through a bad node. It is worth highlighting the difference between the $\Good_1$ and $\G_1$ event: The first one is defined for $\Alg$ having access to $\COND$ oracle, whereas the latter is defined for $\Alg$ having access to $\Ratio$ oracle. In Lemma~\ref{lem:ratio-good}, we claim that $\G_1$ happens with a high probability. Then we show that given the event $\G_1$ happens, $\Good_1$ happens with a high probability. As a consequence, we conclude that $\Good_1$ happens with a high probability, which completes the proof of Lemma~\ref{lem:good1-high-prob}.

\begin{lemma}
    \label{lem:ratio-good}
    If an algorithm $\Alg$ makes $q \le \frac{\log \log n}{100 \log \log \log n}$ adaptive $\Ratio$ queries, then 
    \begin{enumerate}
        \item $\Pr_{x\in_R \IY} \left[\G_1(\Alg, x)\right] \geq 1- o(1)$,
        \item $\Pr_{x\in_R \IN} \left[\G_1(\Alg, x)\right] \geq 1-o(1)$.
    \end{enumerate}
\end{lemma}
\begin{proof}
First, observe that for any $\Ratio$ query, 
Step~\ref{step:ratio-second} (of Definition~\ref{def:ratio-query}) does not depend on the
input distribution (on which the $\Ratio$ query is placed). Thus a randomized algorithm can simulate this step without accessing the input distribution (by picking an atom $V \in \At(U_i)$ with probability $|V|/|U_i|$). Let $\Alg'$ be the new (randomized) algorithm that simulates Step~\ref{step:ratio-second}. One can think of the randomized algorithm $\Alg'$ as a distribution over a set of algorithms $\Alg'_1,\Alg'_2,\cdots$, where each $\Alg'_r$ is an instantiation of $\Alg'$ by fixing internal randomness that is used to simulate Step~\ref{step:ratio-second}. Thus each $\Alg'_r$ can be represented as a decision tree where each node denotes access (of the form either Step~\ref{step:ratio-second} or Step~\ref{step:ratio-third}) to the input distribution. Since the original algorithm $\Alg$ makes $q$ $\Ratio$ queries, each decision tree has a height at most $2q$. Recall, by our assumption, $\Alg$ is a core-adaptive tester. Thus for any node at the $i$-th level, the number of children is $|O_i| + 1$ (follows from Step~\ref{step:ratio-first}), which in turn is upper bounded by $q+1$ (since $\Alg$ makes $q$ $\Ratio$ queries, by the definition of $O_i$ in Definition~\ref{def:core-adaptive}, $|O_i| \le q$). So the number of nodes present in each such decision tree is at most $(q+1)^{2q}$.

We now use the following claim that helps us in bounding the probability of a node being bad (see Definition~\ref{def:good node}) in a decision tree.
\begin{claim}
\label{clm:prob-x-bad-u}
     Irrespective of whether the input $x \in_R \IY$ or $x \in_R \IN$, the probability that a node of a decision tree is bad is at most $ O\left(\frac{2^q (\log \log n)^{20}}{\sqrt{\log n}}\right)$.
\end{claim}
The proof of the above claim is an adaptation of the arguments used in~\cite{acharya2018chasm} and is provided at the end of this subsection. For now, we assume the above claim holds and continue with the proof. Now, by taking a union bound over all the nodes in the decision tree, the probability that at least one of the nodes is bad is at most $(q+1)^{2q} \cdot O\left(\frac{2^q (\log \log n)^{20}}{\sqrt{\log n}}\right)$. Thus when $q \le \frac{\log \log n}{100 \log \log \log n}$, for each $\Alg'_r$,
\[
\Pr_{x \in_R \IY}\left[\Alg'_r\text{ reaches a bad node on input }x\right]=o(1).
\]
Since the randomized algorithm $\Alg'$ can be thought of as a distribution over the set of algorithms $\Alg'_1,\Alg'_2,\cdots$, we get that
\[
\Pr_{x \in_R \IY}\left[\Alg'\text{ reaches a bad node on input }x\right]=o(1).
\]
By the construction of $\Alg'$, for every input $x$, the probability of reaching a bad node by $\Alg$ and $\Alg'$ are the same. Hence, $\Pr_{x\in_R \IY} \left[\G_1(\Alg, x)\right] \geq 1- o(1)$. A similar argument holds for $x \in_R \IN$. This concludes the proof.
\end{proof}

From now, throughout the rest of this subsection, we use $\Alg_{\COND}$ and $\Alg_{\Ratio}$ to denote the core-adaptive tester $\Alg$ (with the corresponding decision tree $R$) having access to the $\COND$ and $\Ratio$ oracle respectively. Next, for any $i \le q$, let us consider the following two distributions: $L_{\COND}(i)$ and $L_{\Ratio}(i)$ are the distributions of nodes at the $i$-th level of the decision tree associated with $\Alg_{\COND}$ and $\Alg_{\Ratio}$ respectively.

Let $\G_2$ be the event of Lemma~\ref{cor:prob-mass-conc}, i.e., for all good nodes $v$ and for all  atoms $A \in \At(\bm{A}_v)$, we have $D_1(A),D_2(A) \in [1-1/\gamma,1+1/\gamma] \frac{\phi_A |A|}{\tau n}$. Now, we upper bound the total variation distance between the distribution over the nodes of $R$ at any depth $i \le q$, for $\Alg_{\COND}$ and $\Alg_{\Ratio}$, conditioned on event $\G_2$.

\begin{lemma}
    \label{lem:tv-small-G1}
    Given the event $\G_1 \wedge \G_2$ happens, for each $i \le q$, the total variation distance between the two distributions $L_{\COND}(i)$ and $L_{\Ratio}(i)$ is at most $4i/\gamma$.
\end{lemma}
\begin{proof}
     Consider an $i \le q$. We denote the total variation distance between the two distributions $L_{\COND}(i)$ and $L_{\Ratio}(i)$ (conditioning on the event $\G_1 \wedge \G_2$) by $\tv(i)$. Note that if a node is bad, then all its descendants are also bad. For any node $v$, let $R(v)$ denote the subtree (of $R$) rooted at the node $v$. Let $R'$ be the tree obtained from $R$ by iteratively pruning $R(v) \setminus v$ (i.e., removing the subtree rooted at $v$ while only keeping the node $v$) for all the bad nodes $v$ in $R$, starting from the lowest level. Observe, in $R'$, no internal node is bad. Furthermore, the probability of reaching a bad node in $R'$ is the same as that in $R$. So it suffices to argue with the distribution of nodes at any level $i$ of the tree $R'$.

Consider any good node $v$ at depth $i$. We upper bound the total variation distance between distributions over children of $v$ for $\Alg_{\COND}$ and $\Alg_{\Ratio}$, conditioned on the fact that the run of the tester is currently at the node $v$. It is not hard to see that if, for each node $v$, the total variation distance between distributions over children of $v$ for $\Alg_{\COND}$ and $\Alg_{\Ratio}$ is upper bounded by $\eta$ (for some $\eta \ge 0$), then
\begin{equation}
    \label{eq:tv-level}
    \tv(i+1) \le \tv(i) + \eta.
\end{equation}

Now, since $\Alg_{\Ratio}$ only differs at Step~\ref{step:ratio-second} (of the $\Ratio$ oracle) with $\Alg_{\COND}$ (see Observation~\ref{obs:cond-new-defn}), it suffices to upper bound the total variation distance for distributions over the atoms of $U_v$. If $U_v$ has no large atom, then $\phi_{U_v}=0$, and thus by Lemma~\ref{cor:prob-mass-conc}, we have $D_1(U_v) = D_2(U_v) = 0$ (conditioned on the event $\G_2$). In this case, the total variation distance is zero. So let us assume that $U_v$ has at least one large atom. Let $V_m$ be the largest atom in $\At(U_v)$ (breaking ties arbitrarily). Hence,
\begin{equation}
    \label{eq:phi-lb}
    \phi_{V_m} \ge \phi.
\end{equation}

Also, by Lemma~\ref{cor:prob-mass-conc}, conditioning on $\G_2$, for all $k \in \{1,2\}$ and $V \in \At(U_v)$, we have
\[
\frac{D_k(V)}{D_k(U_v)} = \theta_V \cdot \frac{\phi_V |V|}{\sum_{W \in \At(U)} \phi_W |W|}
\]
where $\theta_V \in [1-2/\gamma,1+2/\gamma] $. Thus the total variation distance is equal to

\begin{align}
\label{eq:tv-bound}
  &  \sum_{V \in \At(U_v)} \left| \frac{\theta_V \phi_V |V|}{\sum_{W \in \At(U_v)} \phi_W |W|} - \frac{|V|}{|U_v|}\right| \nonumber \\
    &  = \sum_{V \in \At(U_v): |V| < \frac{V_m}{\rho}} \left|\frac{\theta_V \phi_V |V|}{\sum_{W \in \At(U_v)} \phi_W |W|} - \frac{|V|}{|U_v|}\right| +  \sum_{V \in \At(U_v): |V| \ge \frac{V_m}{\rho}}\left|\frac{\theta_V \phi_V |V|}{\sum_{W \in \At(U_v)} \phi_W |W|} - \frac{|V|}{|U_v|}\right|
\end{align}

Now, since $V_m$ is the largest atom in $\At(U_v)$ for any $V$ such that $|V| < \frac{V_m}{\rho}$, both $\frac{\phi_V |V|}{\sum_{W \in \At(U_v)} \phi_W |W|}$ and $\frac{|V|}{|U_v|}$ are in $\left[0,\frac{1}{\rho}\right]$. Since $\At(U_v)$ contains at most $2^q \le \log n$ atoms, 
\begin{equation}
    \label{eq:first-term}
    \sum_{V \in \At(U_v): |V| < \frac{V_m}{\rho}}\left|\frac{\theta_V \phi_V |V|}{\sum_{W \in \At(U_v)} \phi_W |W|} - \frac{|V|}{|U_v|}\right| \le \frac{(1+1/\gamma)\log n}{\rho}.
\end{equation}

Now, it directly follows from the Definition~\ref{def:phi} that
\[
\text{For any }V,W \in \At(U_v) \text{ such that } \frac{|V_m|}{\rho} \le |V|,|W| \le |V_m|,\quad \quad|\phi_V - \phi_W| \in \{0,1\}.
\]
Therefore,
\[
\left(\frac{\phi_{V_m} -1}{\phi_{V_m}} \right)\cdot \frac{|V|}{|U_v|} \le  \frac{\phi_V |V|}{\sum_{W \in \At(U_v)} \phi_W |W|} \le \left(\frac{\phi_{V_m} }{\phi_{V_m} -1}\right) \cdot \frac{|V|}{|U_v|}
\]
and hence
\begin{equation}
    \label{eq:second-term}
    \left|\frac{\theta_V\phi_V |V|}{\sum_{W \in \At(U_v)} \phi_W |W|} - \frac{|V|}{|U_v|}\right| \le \left(\frac{1}{\gamma} + \frac{1}{\phi_{V_m}}\right)\cdot \frac{|V|}{|U_v|} \le \frac{2}{\gamma}\cdot \frac{|V|}{|U_v|}
\end{equation}
where the last inequality follows from Equation~\ref{eq:phi-lb}. Hence, by combining Equation~\ref{eq:tv-bound},~\ref{eq:first-term}, and~\ref{eq:second-term}, the total variation distance is bounded above by 
\[
\frac{(1+1/\gamma)\log n}{\rho} + \frac{2}{\gamma} \le  \frac{4}{\gamma}.
\]
Then by Equation~\ref{eq:tv-level},
\[
\tv(i+1) \le \tv(i) + 4/\gamma.
\]
So we conclude that for each $i \le q$, $\tv(i) \le  4i/\gamma$. 
\end{proof}

Now, we are ready to finish the proof of Lemma~\ref{lem:good1-high-prob}.
\begin{proof}[Proof of Lemma~\ref{lem:good1-high-prob}]
It directly follows from Lemma~\ref{lem:tv-small-G1} that the probability that $\Alg_{\COND}$ reaches a bad node (on the input $x$) is at most 
\[
\sum_{i=1}^q 4i/\gamma + \Pr[\overline{\G_1}] + \Pr[\overline{\G_2}] \le 2q^2/\gamma + \Pr[\overline{\G_1}] + \Pr[\overline{\G_2}].
\]
Then by Lemma~\ref{lem:ratio-good} and Lemma~\ref{cor:prob-mass-conc},
\[
\Pr_{x \in_R \IY}\left[\overline{\Good_1(\Alg,x)}\right] \le 2q^2/\gamma + o(1) \le \delta/6
\]
where the last inequality follows for $q \le \frac{\log \log n}{100 \log \log \log n}$. The same argument also holds for $\Pr_{x \in_R \IN}\left[\overline{\Good_1(\Alg,x)}\right]$, and that concludes the proof.
\end{proof}

So it only remains to prove Claim~\ref{clm:prob-x-bad-u}.
\begin{proof}[Proof of Claim~\ref{clm:prob-x-bad-u}]
Note that $b = 2^\kappa$. Fix a node $v$.  Let 
\begin{align*}
\K_1 =\bigcup_{j \in [\tau],A \in \At(\bm{A}_v)} \left\{\kappa:  \log \left(\frac{n}{\alpha |A| \rho^j}\right) <\kappa < \log \left(\frac{n \alpha }{|A| \rho^j}\right)\right\}
\end{align*}
be all the possible values of $\kappa$ for which the first item of Definition~\ref{def:good node} can get violated. We have
\[
|\K_1| \le \tau \cdot |\At(\bm{A}_v)| \cdot \log \frac{\frac{n \alpha }{|A| \rho^j}}{\frac{n}{\alpha |A| \rho^j}} \le \tau  2^{q+1} \log \alpha.
\]
Let 
\begin{align*}
    \K_2 =\bigcup_{A \in \At(\bm{A}_v)} \left\{\kappa:  \log \left(\frac{n}{\alpha |A| \rho^\tau}\right) <\kappa < \log \left(\frac{n \alpha \rho^\phi}{|A|}\right)\right\}
\end{align*}
 be all the possible values of $\kappa$ for which the second item of Definition~\ref{def:good node} can get violated. We have
 \[
 |\K_2| \le |\At(\bm{A}_v)| \cdot \log \frac{\frac{n \alpha \rho^\phi}{|A|}}{\frac{n}{\alpha |A| \rho^\tau}} \le 2^q (2 \log \alpha + \phi \log \rho).
 \]

Let  

\begin{align*}
\K_3 =\bigcup_{i' \in [i],j \in [\tau]} \left\{\kappa:\frac{n}{\gamma \rho^j}  <  2^\kappa \sum_{A \in At(U_{v_{i'}})} \phi_A |A| < \frac{\gamma n}{\rho^j}\right\}    
\end{align*}
be all the possible values of $\kappa$ for which the third item of Definition~\ref{def:good node} can get violated. 
Observe that 
\begin{align*}
\left|\left\{\kappa:\frac{n}{\gamma \rho^j}  <  2^\kappa \sum_{A \in At(U_{v_{i'}})} \phi_A |A| < \frac{\gamma n}{\rho^j}\right\}\right| \le \left|\left\{\kappa:\frac{n}{\gamma \rho^j}  <  2^\kappa  < \frac{\gamma n}{\rho^j}\right\}\right|.
\end{align*}

So we have $|\K_3| \le 2q \tau \log \gamma$. Therefore,
\[
|\K_1|+|\K_2|+|\K_3| \le \left( 2^{q+1} \log \alpha +  2^q \left(2 \log \alpha + \phi \log \rho\right) + 2q \tau \log \gamma\right).
\]
In both $\YES$ and $\NO$ instances, since $\kappa$ is drawn uniformly from $\{0,1,\dots,\lfloor \frac{\log n}{2} \rfloor \}$, the probability that any given node $v$ is bad is at most
\[
\frac{\left( 2^{q+1} \log \alpha +  2^q \left(2 \log \alpha + \phi \log \rho\right) + 2q \tau \log \gamma\right)}{\log n} \le O\left(\frac{2^q (\log \log n)^{20}}{\sqrt{\log n}}\right).
\]
\end{proof}

\subsection{Proof of Lemma~\ref{lem:if-good-then-lb}}\label{sec:ifGoodthenlb}
 Let $V^{\YES}_i$ be the node reached by $\Alg$ after $i$ queries when the input is drawn from the $\YES$ instance. And let $B^{\YES}(s_i)$ be the bucket in which the $i$th sample belongs. 
Similarly, we define $V^{\NO}_i$ and $B^{\NO}(s_i)$.

Conditioned on the event $\Good$ we will prove that  for any node $v_i$ (at depth $i$) and $(b_1,\dots,b_i) \in [\tau]^i$, 
\begin{align*}
\Pr\left[V^{\YES}_q = v_q, (B^{\YES}(s_1),\dots,B^{\YES}(s_q)) = (b_1,\dots,b_q)\right] & \le \\
\left(1+\frac{100i}{\gamma}\right) & \Pr\left[V^{\NO}_q = v_q, (B^{\NO}(s_1),\dots,B^{\NO}(s_q)) = (b_1,\dots,b_q)\right]
\end{align*}


We prove it by induction on $i$. 
Let the parent of the node $v_i$ be $v_{i-1}$. For brevity, we use $\mcB^{\YES}(\bm{S}_i)$ for $\left(B^{\YES}(s_1),\dots,B^{\YES}(s_i)\right)$ and $\bm{B}_i$ for $(b_1,\dots,b_i)$. 
\begin{align*}
   & \Pr\left[V^{\YES}_i = v_i, \mcB^{\YES}(\bm{S}_i) = \bm{B}_i\right] \\
  = & \Pr\left[V^{\YES}_i = v_i, V^{\YES}_{i-1} = v_{i-1}, \mcB^{\YES}(\bm{S}_i) = \bm{B}_i\right]\\
     = & \Pr\left[V^{\YES}_i = v_i,B^{\YES}(s_i) = b_i|V^{\YES}_{i-1} = v_{i-1},\mcB^{\YES}(\bm{S}_{i-1}) = \bm{B}_{i-1}] \cdot \Pr[V^{\YES}_{i-1} = v_{i-1},\mcB^{\YES}(\bm{S}_{i-1}) = \bm{B}_{i-1}\right]\\
     \le & \Pr\left[V^{\YES}_i = v_i,B^{\YES}(s_i) = b_i|V^{\YES}_{i-1} = v_{i-1},\mcB^{\YES}(\bm{S}_{i-1}) = \bm{B}_{i-1}\right] \cdot \\
   & \hspace{8cm} \left(1+\frac{100(i-1)}{\gamma}\right) \Pr\left[V^{\NO}_{i-1} = v_{i-1},\mcB^{\NO}(\bm{S}_{i-1}) = \bm{B}_{i-1}\right]\\
   = & \left(1+\frac{100(i-1)}{\gamma}\right) \Pr\left[V^{\YES}_i = v_i,B^{\YES}(s_i) = b_i\mid V^{\YES}_{i-1} = v_{i-1},\mcB^{\YES}(\bm{S}_{i-1})= \bm{B}_{i-1}\right]  \cdot \\
   & \hspace{9cm} \Pr\left[V^{\NO}_{i-1} = v_{i-1},\mcB^{\NO}(\bm{S}_{i-1}) = \bm{B}_{i-1}\right] 
\end{align*}
Fixing the node $v_{i-1}$ fixes  the next query $\COND_k(O_{v_{i-1}} \cup \{U_{v_{i-1}}\})$ irrespective of the $\YES$ or $\NO$ instances. Since we are conditioning on the event $\Good_3$, the unique heaviest element  $e \in O_{v_{i-1}} \cup \{U_{v_{i-1}}\}$ will be the outcome of this query  and is same for both instances. 
We have two cases to consider.
\begin{enumerate}
    \item If $e \in O_{v_{i-1}}$ then the value of $V^{\YES}_i$ and $V^{\NO}_i$  is same which is (with probability $1$) the child of $v^{i-1}$  with the corresponding edge labeled by $e$. Obviously, there is no new seen element in this case, i.e., as $e \in O_{v_{i-1}}$ so we have  $s_i = s_j$ for some $j \in [i-1]$ and hence $B^{\YES}(s_i) = B^{\NO}(s_i) = b_j$. Hence, in this case, we have  
    \begin{align*}
    \Pr\left[V^{\YES}_i = v_i,B^{\YES}(s_i) = b_i\mid V^{\YES}_{i-1} = v_{i-1}, \right. & \left. \mcB^{\YES}(\bm{S}_{i-1}) = \bm{B}_{i-1}\right] \\
   =  \Pr &\left[V^{\NO}_i = v_i,B^{\NO}(s_i) = b_i\mid V^{\NO}_{i-1} = v_{i-1},\mcB^{\NO}(\bm{S}_{i-1}) = \bm{B}_{i-1}\right]
    \end{align*}

    \item Now  consider  $e = \{U_{v_{i-1}}\}$.  In this case, by Lemma~\ref{cor:prob-mass-conc} in both $\YES$ and $\NO$ instances, each atom $W \in \At(U_{v_{i-1}})$ is picked with probability whose value is in $\left[1-1/\gamma,1+1/\gamma\right]\frac{|W|}{|U_{v_{i-1}}|}$. Note that each atom $U$ in one to one manner corresponds to a child of the node $v_{i-1}$. Thus we have 
\begin{align*}
\Pr\left[V^{\YES}_i = v_i \mid V^{\YES}_{i-1} = v_{i-1},\mcB^{\YES}(\bm{S}_{i-1}) = \bm{B}_{i-1}\right] & \\
\le & \left( \frac{1+1/\gamma}{1- 1/\gamma} \right) \Pr\left[V^{\NO}_i = v_i \mid V^{\NO}_{i-1} = v_{i-1},\mcB^{\NO}(\bm{S}_{i-1})= \bm{B}_{i-1}\right]
\end{align*}

 From Lemma~\ref{lem:phi}, we have for all $k \in \{1,2\}$, $D_k(W \cap B_j) = 0$ if $j \le r- \phi_W$ and $\left[1-\frac{1}{\gamma},1+\frac{1}{\gamma}\right] \frac{|W|}{\tau n}$ if  $j \ge r- \phi_W + 1$. Therefore, for any $k \in \{1,2\}$, the index of the bucket $b(s_i)$, when $s_i \sim \COND_k(W)$, takes any particular value in $\{r-\phi_W +1,r-\phi_W+2,\dots,r\}$ with probability  $\left[1-\frac{1}{\gamma},1+\frac{1}{\gamma}\right] \frac{1}{\phi_W}$ and any value in $\{1,2,\dots,r - \phi_W\}$ with probability $0$.
 \begin{align*}
 \Pr\left[B^{\YES}(s_i) = b_i\mid V^{\YES}_i = v_i,V^{\YES}_{i-1} = v_{i-1},\mcB^{\YES}(\bm{S}_{i-1}) = \bm{B}_{i-1}\right] & \\
 \le  \left(\frac{1+1/\gamma}{1- 1/\gamma} \right)  \Pr\left[B^{\NO}(s_i) = b_i\mid  V^{\NO}_i =\right. & \left.  v_i,V^{\NO}_{i-1} = v_{i-1},\mcB^{\NO}(\bm{S}_{i-1}) = \bm{B}_{i-1}\right]
\end{align*}

Hence, in this case, we have 
 \begin{align*}
   \Pr\left[V^{\YES}_i = v_i,B^{\YES}(s_i) = b_i\mid V^{\YES}_{i-1} = v_{i-1},\mcB^{\YES}(\bm{S}_{i-1}) = \bm{B}_{i-1}\right] & \\
    \le \left(\frac{1+1/\gamma}{1-1/\gamma}\right)^2  \Pr\left[V^{\NO}_i = v_i,B^{\NO}(s_i) = b_i\mid V^{\NO}_{i-1} \right. & \left. = v_{i-1},\mcB^{\NO}(\bm{S}_{i-1}) = \bm{B}_{i-1}\right]
    \end{align*}
\end{enumerate}  

Since $\left(\frac{1+1/\gamma}{1-1/\gamma}\right)^2 \left(1+100(i-1)/\gamma\right) \le \left(1+100i/\gamma\right)$, the induction hypothesis is true for $i$ as well.

    
    Fixing any $v_q$ and summing the inequality in for all possible values of $(b_1,\dots,b_q)$, we get the result.


\section{ $\Omega(\log \log n)$ lower bound for testing Label Invariant Property}\label{sec:label}

In this section, we prove Theorem~\ref{thm:main-odd-even}. We will prove that there is a label invariant property testing which takes $\Tilde{\Omega}(\log\log n)$ $\COND$ queries. CFGM~\cite{chakraborty2016power} 
defined a label invariant property (called Even uniblock property) and show a lower bound of $\Omega(\sqrt{\log \log n})$  on the query complexity. 
We will improve this lower bound to $\Tilde{\Omega}(\log \log n)$. 

\paragraph*{Even uniblock property:}  A distribution on $[n]$ is called even uniblock if and only if it is uniform over some subset $U \subseteq [n]$ of size $2^{2\kappa}$ for some  $\frac{\log n}{8} \le \kappa \le \frac{3 \log n}{8}$. Note that a distribution $D$ on $[n]$ is said to be uniform  on  set $S \subseteq [n]$ if we have $D(i) = 1/|S|$ for all $i \in S$ and $0$ otherwise.
\paragraph*{Odd uniblock   property:}  A distribution on $[n]$ is called odd uniblock if and only if it is uniform over some subset $U \subseteq [n]$ of size $2^{2\kappa+1}$ for some  $\frac{\log n}{8} \le \kappa \le \frac{3 \log n}{8}$. Note that a distribution $D$ on $[n]$ is said to be uniform  on  set $S \subseteq [n]$ if we have $D(i) = 1/|S|$ for all $i \in S$ and $0$ otherwise.

\begin{observation}[\cite{chakraborty2016power}]
For any distribution $D_e$ satisfying even uniblock property and any distribution $D_o$ satisfying odd uniblock property, we have 
$d_{TV} (D_e, D_o) \ge \frac{1}{2}$.
\end{observation}

To prove Theorem~\ref{thm:main-odd-even} we will use the same hard instances (explained in Section~\ref{sec:labelhard}) as defined in \cite{chakraborty2016power}. Then, in Section~\ref{sec:labelgood}, we define our $\Good$ event, and in Section~\ref{sec:labelgoodprob}, we prove (using the $\Ratio$ model) the $\Good$ event happens with high probability. Finally, in Section~\ref{sec:labelfinal} we prove Theorem~\ref{thm:main-odd-even} using Theorem~\ref{thm:extyao}.

\subsection{Distributions over Even uniblock and Odd uniblock properties}\label{sec:labelhard}
We now consider two distributions  $\IE$ and $\IO$,  over distributions satisfying even uniblock property and  over distributions satisfying odd uniblock property, respectively. Note that these are the same distributions considered in~\cite{chakraborty2016power}.
\begin{enumerate}
    \item Uniformly choose  an integer $\kappa$   such that $\frac{\log n}{8} \le \kappa \le \frac{3 \log n}{8}$.
    \item  Uniformly pick  a set $S_e$ of size  $2^{2\kappa}$ and a set $S_o$ of size $2^{2\kappa+1}$.
    \item  The distribution $\IE$ is a  uniform  distribution on $S_e$ while $\IO$ is a  uniform distribution on $S_o$
\end{enumerate}


\subsection{The $\Good$ event}\label{sec:labelgood}




\begin{definition}[\cite{chakraborty2016power}]
 We call   a number $b$  large with respect to $S_e$ if  $\frac{b |S_e|}{n} \ge 2^{\sqrt{\log n}}$ and small with respect to $S_e$ if 
 $\frac{b |S_e|}{n} < \frac{1}{2^{\sqrt{\log n}}}$. We have an analogous definition for $S_o$. Note that $|S_e| = 2^{2 \kappa}$ and $|S_o| = 2^{2 \kappa +1}$.
\end{definition}

\begin{definition}
\label{def:good node-oddeven}
    A node $v$ of the decision tree $R$ (recall that $\bm{A}_v = (A_{v_0},\dots, A_{v_i})$ is the set of all queries made in the path from the root to the node $v$ in the decision tree.) is called good if for all atoms $A \in At(\bm{A}_v)$, $|A|$ is  large with respect to both $S_e$ and $S_o$ or small with respect to both $S_e$ and $S_o$.
 A node $v$ is called bad  if it is not good.
\end{definition}

We now show the following lemma that helps us in bounding the probability of a node being bad in a decision tree. The proof of the following lemma is an adaptation of the arguments used in~\cite{chakraborty2016power}.
\begin{lemma}
\label{clm:prob-x-bad-u-odd-even}
     Irrespective of whether the input $x \in_R \IE$ or $x \in_R \IO$, the probability that a node of a decision tree is bad is at most $\frac{8 \cdot 2^q}{\sqrt{\log n}}$.
\end{lemma}
\begin{proof}
    An atom $A$ is neither large nor small with respect to $S_e$ if $n 2^{-\sqrt{\log n}} < |A||S_e| < n 2^{\sqrt{\log n}}$ where $|S_e| = 2^{2\kappa}$ and $\kappa$ is chosen uniformly such that $\frac{\log n}{8} \le \kappa \le \frac{3 \log n}{8}$ . Therefore, for a fixed $|A|$, there are at most $\sqrt{\log n}$ values of $\kappa$, which will make it neither large nor small with respect to $S_e$. So there are at most $2 \sqrt{\log n}$ values of $\kappa$, which will make it neither large nor small with respect to  both $S_e$ and $S_o$.
 Since the range of $\kappa$ is $\log n/4$ (in both $\mu_e$ and $\mu_o$), with probability at most $\frac{8}{\sqrt{\log n}},$ the atom $A$ will be neither large nor small. Since there are at most $2^q$ atoms corresponding to a node, the statement of the lemma follows.
\end{proof}

From now on, fix
\[
\beta = 2^{\frac{\sqrt{\log n}}{4}}.
\]

\begin{lemma}\label{lem:prob-conc-odd-even}
 With probability at least $1-o(1)$,   for all good nodes $v$ and for all  atoms $A \in \At(\bm{A}_v)$:
 \begin{enumerate}
     \item If $|A|$ is large, we have  $|A \cap S_e| \in [1-1/\beta,1+1/\beta] \frac{|A| |S_e|}{n}$ and $|A \cap S_o| \in [1-1/\beta,1+1/\beta] \frac{|A| |S_o|}{n}$.
     \item If $|A|$ is small, we have $|A \cap S_e| = |A \cap S_o| =0$. 
 \end{enumerate}
\end{lemma}
\begin{proof}
     If  $\E\left[|S_e \cap A|\right] = \frac{|A||S_e|}{n} \le \frac{1}{2^{\sqrt{\log n}}}$ for any atom $A$  then by Markov's inequality $\Pr[|A \cap S_e| > 0] \le \frac{1}{2^{\sqrt{\log n}}}$.  Taking a union bound over all the atoms corresponding to all the good nodes (which can be at most $2^{q+q^2} \le (\log n)^{\log \log n}$), with probability at least $1 - o(1)$, 
    \begin{equation}
        \mbox{ for any small atom $A$, we have $|S_e \cap A| = 0$}.
    \end{equation}

    Now note that for any large atom,
    $\E\left[|S_e \cap A|\right] = \frac{|A||S_e|}{n} \ge 2^{\sqrt{\log n}} $. Therefore, by Chernoff bound,
    \[
    \Pr\left[|A \cap S_e| \not \in \left[1-\frac{1}{\beta}, 1+ \frac{1}{\beta}\right] \frac{|A||S_e|}{n}\right] \le O\left(exp\left(-2^{\sqrt{\log n}/\beta^2}\right)\right) \le O\left(exp\left(-2^{\sqrt{\log n}/2}\right)\right).
    \]
The proof analogously holds for $S_o$.
\end{proof}

\begin{definition}
\label{def:good-uniblock}
    The event $\Good(\Alg, x)$ is defined as ``the run of the algorithm $\Alg$ on input $x$ does not pass through a bad node."
\end{definition}

\subsection{The event $\Good(\Alg, x)$ happens with high probability}\label{sec:labelgoodprob}
In this subsection, we will prove the following lemma.
\begin{lemma}\label{lem:Good-odd-even} For any  core-adaptive tester $\Alg$ that makes $q \le \frac{\log \log n}{100 \log \log \log n}$  queries,
     \begin{enumerate}
         \item $\Pr_{x \in_R \IE}[\Good(\Alg, x)] \geq 1 - o(1)$
        \item $\Pr_{x \in_R  \IO}[\Good(\Alg, x)] \geq 1- o(1)$
     \end{enumerate}
\end{lemma}
\subsubsection{Proof for $\Ratio$ oracle}
Let us consider the event $\G_1$ that denotes the run of an algorithm $\Alg$ having access to $\Ratio$ oracle, on input $x$ does not pass through a bad node. In Lemma~\ref{lem:ratio-good-odd-even}, we claim that $\G_1$ happens with a high probability. Then we show that given the event $\G_1$ happens, $\Good$ happens with a high probability. As a consequence, we conclude that $\Good$ happens with a high probability, which completes the proof of Lemma~\ref{lem:Good-odd-even}.

\begin{lemma}
    \label{lem:ratio-good-odd-even}
    If an algorithm $\Alg$ makes $q \le \frac{\log \log n}{100 \log \log \log n}$ adaptive $\Ratio$ queries, then 
    \begin{enumerate}
        \item $\Pr_{x \in_R \IE} \left[\G_1(\Alg, x)\right] \geq 1- o(1)$,
        \item $\Pr_{x \in_R \IO} \left[\G_1(\Alg, x)\right] \geq 1-o(1)$.
    \end{enumerate}
\end{lemma}
\begin{proof}
First, observe that for any $\Ratio$ query, Step~\ref{step:ratio-second} (of Definition~\ref{def:ratio-query}) does not depend on the input distribution (on which the $\Ratio$ query is placed). Thus, a randomized algorithm can simulate this step without accessing the input distribution (by picking an atom $V \in \At(U_i)$ with probability $|V|/|U_i|$). Let $\Alg'$ be the new (randomized) algorithm that simulates Step~\ref{step:ratio-second}. One can think of the randomized algorithm $\Alg'$ as a distribution over a set of algorithms $\Alg'_1,\Alg'_2,\cdots$, where each $\Alg'_r$ is an instantiation of $\Alg'$ by fixing internal randomness that is used to simulate Step~\ref{step:ratio-second}. Thus, each $\Alg'_r$ can be represented as a decision tree where each node denotes access (of the form either Step~\ref{step:ratio-second} or Step~\ref{step:ratio-third}) to the input distribution. Since the original algorithm $\Alg$ makes $q$ $\Ratio$ queries, each decision tree has a height at most $2q$. Recall, by our assumption, $\Alg$ is a core-adaptive tester. Thus for any node at the $i$-th level, the number of children is $|O_i| + 1$ (follows from Step~\ref{step:ratio-first}), which in turn is upper bounded by $q+1$ (since $\Alg$ makes $q$ $\Ratio$ queries, by the definition of $O_i$ in Definition~\ref{def:core-adaptive}, $|O_i| \le q$). So the number of nodes present in each such decision tree is at most $(q+1)^{2q}$.

  Now, by Lemma~\ref{clm:prob-x-bad-u-odd-even} and  taking a union bound over all the nodes in the decision tree, the probability that at least one of the nodes is bad is at most $\frac{8 \cdot 2^q \cdot (q+1)^{2q}}{\sqrt{\log n}}$. Thus when $q \le \frac{\log \log n}{100 \log \log \log n}$, for each $\Alg'_r$,
\[
\Pr_{x \in_R \IE}\left[\Alg'_r\text{ reaches a bad node on input }x\right]=o(1).
\]
Since the randomized algorithm $\Alg'$ can be thought of as a distribution over the set of algorithms $\Alg'_1,\Alg'_2,\cdots$, we get that
\[
\Pr_{x \in_R \IE}\left[\Alg'\text{ reaches a bad node on input }x\right]=o(1).
\]
By the construction of $\Alg'$, for every input $x$, the probability of reaching a bad node by $\Alg$ and $\Alg'$ are the same. Hence, $\Pr_{x \in_R \IE} \left[\G_1(\Alg, x)\right] \geq 1- o(1)$. A similar argument holds for $x \in_R \IO$. This concludes the proof.
\end{proof}

\subsubsection{The total variation distance is small for testers given access to $\COND$ and $\Ratio$ oracle }

From now on, throughout the rest of this subsection, we use $\Alg_{\COND}$ and $\Alg_{\Ratio}$ to denote the core-adaptive tester $\Alg$ (with the corresponding decision tree $R$) having access to the $\COND$ and $\Ratio$ oracle respectively. Next, for any $i \le q$, let us consider the following two distributions: $L_{\COND}(i)$ and $L_{\Ratio}(i)$ are the distributions of nodes at the $i$-th level of the decision tree associated with $\Alg_{\COND}$ and $\Alg_{\Ratio}$ respectively.

Let $\G_2$ be the event of Lemma~\ref{lem:prob-conc-odd-even}, i.e., for all good nodes $v$ and for all  atoms $A \in \At(\bm{A}_v)$, we have $|A \cap S_e| \in [1-1/\beta,1+1/\beta] \frac{|A| |S_e|}{ n}$ and $|A \cap S_o| \in [1-1/\beta,1+1/\beta] \frac{|A| |S_o|}{ n}$. Now, we upper bound the total variation distance between the distribution over the nodes of $R$ at any depth $i \le q$, for $\Alg_{\COND}$ and $\Alg_{\Ratio}$, conditioned on the event $\G_1 \wedge \G_2$.

\begin{lemma}
    \label{lem:tv-small-G1-odd-even}
    Given the event $\G_1 \wedge \G_2$ happens, for each $i \le q$, the total variation distance between the two distributions $L_{\COND}(i)$ and $L_{\Ratio}(i)$ is at most $2i/\beta$.
\end{lemma}
\begin{proof}
     Consider an $i \le q$. We denote the total variation distance between the two distributions $L_{\COND}(i)$ and $L_{\Ratio}(i)$ (conditioning on the event $\G_1 \wedge \G_2$) by $\tv(i)$. Note that if a node is bad, then all its descendants are also bad. For any node $v$, let $R(v)$ denote the subtree (of $R$) rooted at the node $v$. Let $R'$ be the tree obtained from $R$ by iteratively pruning $R(v) \setminus v$ (i.e., removing the subtree rooted at $v$ while only keeping the node $v$) for all the bad nodes $v$ in $R$, starting from the lowest level. Observe in $R'$ no internal node is bad. Furthermore, the probability of reaching a bad node in $R'$ is the same as that in $R$. So, it suffices to argue with the distribution of nodes at any level $i$ of the tree $R'$.

Consider any good node $v$ at depth $i$. We upper bound the total variation distance between distributions over children of $v$ for $\Alg_{\COND}$ and $\Alg_{\Ratio}$, conditioned on the fact that the run of the tester is currently at the node $v$. It is not hard to see that if, for each node $v$, the total variation distance between distributions over children of $v$ for $\Alg_{\COND}$ and $\Alg_{\Ratio}$ is upper bounded by $\eta$ (for some $\eta \ge 0$), then
\begin{equation}
    \label{eq:tv-level-odd-even}
    \tv(i+1) \le \tv(i) + \eta.
\end{equation}

Now, since $\Alg_{\Ratio}$ only differs at Step~\ref{step:ratio-second} (of the $\Ratio$ oracle) with $\Alg_{\COND}$, it suffices to upper bound the total variation distance for distributions over the atoms of $U_v$. If $U_v$ has no large atom,  by Lemma~\ref{lem:prob-conc-odd-even}, we have $|U_v \cap S_e| = |U_v \cap S_o| = 0$ (conditioned on the event $\G_2$). In this case, the total variation distance is zero. So, let us assume that $U_v$ has at least one large atom.

Further, as we are conditioning on $\G_2$,  we have both  $\frac{D_e(V)}{D_e(U_v)}$ and $\frac{D_o(V)}{D_o(U_v)}$ $= \theta_V \frac{|V|}{|U_v|}$ where $\theta_V \in [1-2/\beta,1+2/\beta] $. In this case, the total variation distance is at most 
\begin{align*}
  &  \sum_{V \in At(U)}\left| \frac{\theta_V  |V|}{ |U|} - \frac{|V|}{|U|}\right|
      \le \frac{2}{\beta}.
\end{align*}

Then by Equation~\ref{eq:tv-level-odd-even},
\[
\tv(i+1) \le \tv(i) + 2/\beta.
\]
So we conclude that for each $i \le q$, $\tv(i) \le  2i/\beta$. 
\end{proof}

\subsubsection{$\Good$ happens with high probability}
\begin{proof}[Proof of Lemma~\ref{lem:Good-odd-even}]
It directly follows from Lemma~\ref{lem:tv-small-G1-odd-even} that the probability that $\Alg_{\COND}$ reaches a bad node (on the input $x$) is at most 
\[
\sum_{i=1}^q 2i/\beta + \Pr[\overline{\G_1}] + \Pr[\overline{\G_2}] \le q^2/\beta + \Pr[\overline{\G_1}] + \Pr[\overline{\G_2}].
\]
Then by Lemma~\ref{lem:ratio-good-odd-even} and Lemma~\ref{lem:prob-conc-odd-even},
\[
\Pr_{x \in_R \IE}\left[\overline{\Good(\Alg,x)}\right] \le q^2/\beta + o(1) = o(1)
\]
where the last inequality follows for $q \le \frac{\log \log n}{100 \log \log \log n}$. The same argument also holds for $\Pr_{x \in_R \IO}\left[\overline{\Good(\Alg,x)}\right]$, and that concludes the proof.
\end{proof}

\subsection{Proof of Theorem~\ref{thm:main-odd-even}}\label{sec:labelfinal}

\begin{lemma}[\cite{chakraborty2016power}] \label{lem:if-good-then-small-tv-odd-even}
   Given the event $\Good$ happens, consider the resulting distributions over the set of leaves reached by the algorithm. These two distributions, under $\IE$ compared to under $\IO$, are at most $\frac{2^{3q+1}}{\sqrt{\log n}}$ apart (in the total variation distance) from each other.
\end{lemma}
We want to point out that in~\cite{chakraborty2016power}, the event $\Good$ in the above lemma stands for  ``none of the nodes in the decision tree is bad". However, the proof only uses the fact that none of the nodes encountered by the run of the algorithm is bad. Hence, the lemma remains valid for our definition of the event $\Good$ (as in Definition~\ref{def:good-uniblock}) as well.

Now consider a decision tree of the $\Alg$ and feed to it either $\mu_e$ or $\mu_o$. Unless the queries made by $\Alg$ is $\Tilde{\Omega}(\log \log n)$, in both cases, the event $\Good$ happens with $1-o(1)$ probability (Lemma~\ref{lem:Good-odd-even}). Further, conditioned on the event $\Good$, the total variation distance between the resulting distribution over the leaves is at most $o(1)$ (Lemma~\ref{lem:if-good-then-small-tv-odd-even}).  Hence, the (unconditional) total variation distance between the resulting distribution over the leaves is at most $o(1) +o(1) = o(1)$. This means the $\Alg$ cannot distinguish between $\mu_e$ and $\mu_o$ (unless $q = \Tilde{\Omega}(\log \log n)$).

\section{Conclusion}\label{sec:conclusion}

In this paper, we introduce the $\Ratio$ model and show how this model can be used to obtain improved lower bounds for the seemingly stronger $\COND$ model. The concept of the core-adaptive tester was crucially used in previous works to prove the previous best lower bounds. In our paper, we kind of show that the core-adaptive tester with $\COND$ queries is similar to the $\Ratio$ samplers. In this context, we want to leave with an important question:

\begin{center}
\textbf{Are core-adaptive testers equivalent to a further restricted class of testers? }
\end{center}

We believe this question is important for understanding exactly where the power of $\COND$ lies. Suppose we restrict the definition of the configuration (Definition~\ref{def:configuration}) of the $i$ sample by removing the Condition~\ref{step:confg-2}, i.e.,  now the configuration $c'_i$ is the information whether $s_i = s_j$ for which $j \le i-1$. Consider a restricted class of core-adaptive testers  where, for any $i \in [q]$, the $i$-th query $A_i$ is
\begin{enumerate}
    \item  Either of the form $\COND_{D_k}(U)$ where $U$ is the set of unseen elements (and in this case, we get a sample $s_i \sim \COND_{D_k}(U)$),
    \item Or of the form $\COND_{D_k}(O_i)$  where as before $O_i$ is some  subset of previously seen elements. 
\end{enumerate}
Further, the next query set $A_{i+1}$ is decided by the configurations of the previous samples.

In simplest words, at any step, the tester from this restricted class has the power to (i) sample $s \sim \COND(U)$ for any set of unseen elements $U$, (ii) can determine the value of $j \le i-1$ such that  $s_j = s$ where   $s \sim \COND(O)$. One can also say that this class of testers does not consider the unique atoms to which the samples belong for any decision. It is easy to see that this restricted class of testers is the same as the tester when given access to $\Ratio$ queries.  

This seems like a big restriction, and indeed the size of the decision tree for testers in this class is $q^q$ as compared to $2^{q^2}$ before. Interestingly, to the best of our knowledge, all the testers for any label invariant property in the literature belong to this restricted class. Further, our work shows that the power of testers from this restricted class is the same as that of general core-adaptive testers in the context of equivalence and even uniblock testing. This raises an important question.

\textit{
For what subset of label invariant properties does this restricted class of testers have the same power as general core-adaptive testers?}

\bibliography{references}
\end{document}